\documentclass{llncs}

\usepackage[margin=1in]{geometry}

\usepackage{amsmath}
\usepackage{amsfonts}
\usepackage{amssymb}

\usepackage[utf8]{inputenc}
\usepackage[T1]{fontenc}

\usepackage{times}

\usepackage{hyperref}
\usepackage{verbatim}
\usepackage{enumerate}
\usepackage{enumitem}

\usepackage{cite}

\usepackage{nicefrac}

\usepackage{algorithm}
\usepackage{algorithmicx}
\usepackage{algpseudocode}

\usepackage{pdfsync}
\makeatletter
\newtheorem{rep@theorem}{\rep@title} \newcommand{\newreptheorem}[2]{%
\newenvironment{rep#1}[1]{%
\def\rep@title{\bf #2 \ref{##1} }%
\begin{rep@theorem} }%
{\end{rep@theorem} } }
\makeatother

\newreptheorem{theorem}{Theorem}
\newreptheorem{corollary}{Corollary}
\newreptheorem{lemma}{Lemma}

\newcommand{\bigO}[1]{\mathcal{O}\left(#1\right)}
\newcommand{\bigOm}[1]{{\Omega}\left(#1\right)}
\newcommand{\bigTheta}[1]{{\Theta}\left(#1\right)}

\newcommand{\cP}{\mathcal{P}}
\renewcommand{\Pr}[2]{\operatorname*{\mathbf{Pr}}_{#1}\left(#2\right) }
\newcommand{\PrD}[1]{\mathbf{P}_{#1}}
\newcommand{\ExV}[2]{\operatorname*{\mathbf{E}}_{#1}\left(#2\right)} 

\newcommand{\supp}[1]{\mathrm{supp}\left(#1\right)}

\newcommand{\I}[1]{\mathbf{1}_{#1}}

\newcommand{\poly}[1]{\mathrm{poly}\left(#1\right)}

\newcommand{\Hmin}[1]{\mathbf{H}_{\infty}\left( #1 \right)}

\newcommand{\Hmtr}[2]{\mathbf{H}^{\textup{Metric}, #2}_{\infty}\left( #1 \right)}
\newcommand{\Hsmooth}[2]{\mathbf{H}^{#2}_{\infty}\left( #1 \right)}

\newcommand{\cD}{\mathcal{D}}
\newcommand{\cY}{\mathcal{Y}}

\title{A Convex Analysis Approach to Computational Entropy}
\author{Maciej Sk\'{o}rski}
\institute{University of Warsaw}

\begin{document}

\maketitle

\begin{abstract}

\noindent
This paper studies the notion of computational entropy.  
Using techniques from convex optimization, we investigate the following problems: 
\begin{enumerate}
\item Can we derandomize the computational entropy? More precisely, for the computational entropy, what is the real difference in security defined using the three important classes of circuits: deterministic boolean, deterministic real valued, or (the most powerful) randomized ones? 
\item How large the difference in the computational entropy for an unbounded versus efficient adversary can be? 
\item Can we obtain useful, simpler characterizations for the computational entropy?
\end{enumerate}
The first question was answered affirmatively for the most important notion of HILL entropy but
was open for the metric-type computational entropy, widely used in the leakage-resilent cryptography. 
In this case, we show that the answer depends on what is the underlying variant of the information-theoretic entropy in the definition of the metric entropy.  More precisely, the answer is \emph{negative} for the commonly used min-entropy based computational entropy. Surprisingly, we show that for all other Renyi entropies the answer is \emph{positive} - security given by unbounded deterministic circuits can be still much worse than that guaranteed by efficient randomized circuits.

In the second problem, we obtain some lower-bound type results. Especially, considering conditional computational entropy for two random variables $X\in \{0,1\}^n$ and $Z\in\{0,1\}^m$, we show that 
even if the security parameters are exponential in $n+m$, the ammount of entropy can be still noticeably higher than that seen by unbounded adversary. Also, for a fixed distribution, decreasing the security parameters by a factor $2^C$, can result in increasing the entropy by $C$ bits, which agrees with intuition.

Studying the third problem, we derive a series of lemmas giving a characterization of the metric entropy for various definitions. As an example of application, we give extremely simple proofs of leakage lemmas, being a central tool in the leakage-resilent cryptography.
\end{abstract}

\thispagestyle{empty}

\newpage

\pagestyle{headings}
\setcounter{page}{1}

\section{Introduction}
Entropy is the fundamental concept on which information-theory is founded.  Since its introduction by Shannon  \cite{Shannon1948} the definition of entropy has been generalized in many ways, including the computational variants of this notion (introduced in \cite{Yao1982} and \cite{HILL99}), which turn out to be very useful in the computational complexity theory and cryptography.


There are at least three important and different natural approaches to define computational entropy: the first one based on compressibility (``Yao entropy''), the second one based on the notion of unpredictability (``unpredictability entropy'') and the other one based on the concept of computational indistinguishability (``Metric and HILL entropy''). The relationships between these notions were studied first by Barak et al.\ in \cite{Barak2003}; the reader might also wish to refer to \cite{Reyzin2011} for a survey. In the recent years probably the most popular computational entropy variants  were the Metric and the Hill entropies.  This is partly due to the fact that this notion is often used by authors studying leakage-related problems (Dziembowski and Pietrzak \cite{Dziembowski2008}, Reingold et al. \cite{Reingold2008}, Reyzin and Fuller \cite{Fuller2011}, Kai-Min Chung et al.\ \cite{Chung2011}, Krenn et al.\ \cite{Pietrzak2012}). 
The second important reason is that applying an extractor to a random variable having high HILL Entropy (or even Metric Entropy), one obtains a pseudorandom distribution \cite{Fuller2011}. 

\subsection{Computational Entropy issues}
A major difficulty with the use computational entropy is that it can be defined in many ways, depending on particular application, that often seem to be nonequivalent or not to admit a simple proof. Most often the differences come from the usage of different classes of distinguishers or because there is no standard way of defining conditional computational entropy. As a consequence, for many results in this area we do not know whether they are true if a small change in the definition is made. An example of such a situation is the notion of Metric$^{*}$ Entropy introduced in \cite{Dziembowski2008} and generalized  in \cite{Fuller2011}, reflecting in both cases the problem with determinining what are the relationships between Metric Entropy computed against different classes of distinsguishers: boolean deterministic, $[0,1]$-valued or boolean randomized ones (for the HILL Entropy it is easy to show that all these classes are equivalent \cite{Fuller2011}). 

The another important issue is a very useful estimate used in the leakage-resilent cryptography, called
the ``leakage chain rule", provable for restricted types of conditional computational entropy but known to be false in general \cite{Pietrzak2012}. Yet another important topic is existence of a simple characterization for the Metric Entropy in special cases.  Besides of being of independent interest, such characterizations can have surprisingly powerful applications (cf.\ Section 7 in \cite{Barak2003}).  Thus, although a lot has been done, it seems that systematization of definitions and studying relations between different variants of entropy (even in most often used circuits model) is still needed. 
Our motivation is to contribute to this task, focusing on indistinguishability based computational entropy through this paper.

\subsection{Our techniques}
Our main technique is a novel and interesting observation that the concept of the computational entropy is strictly related to the \emph{separating convex sets} problem. This approach turns out to be especially usefull for the metric-type of computational entropy. The ``extreme" distributions that satisfy the metric-entropy constraints turn out to be
indeed \emph{extreme points} and allow us to apply the powerful machinery of convex analysis. Especially we show that such problems as comparing the security of the metric entropy in different models of an adversary, are deeply dependent on \emph{the geometry of certain convex sets}. We believe that this approach can be of independent interest.

\subsection{Our results and the organization of the paper}
The remainder of the paper is organized as follows. In Section 2 we give some basic notations and introduce definitions
of Computational Entropy. In Section 3, by techniques similar to these used in results related to $\delta$-hard functions, we show a separation between Computational Entropy and Smooth Entropy (which can be viewed as comparing Computational Entropy seen by a bounded and by an unbounded adversary). In Section 4, by solving convex optimization problems, we obtain explicit characterizations of most interesting generalizations of metric entropy. As some of examples of application we reprove the classical relationship between R\'{e}nyi Entropy for different orders and also give an extremely  short proof of the `leakage lemma" and the `leakage chain rule'  for so called relaxed entropy. Section 5 deals with the problem of comparing Computational Metric Entropy for  
different classes of distinguishers used in the definition. We show that it can happen, that the deterministic unbounded adversary is \emph{much more weaker} than for the efficient randomized one. Surprisingly this is not the case of the most popular metric min-entropy. Especially, we construct a random variable $X\in\{0,1\}^n$ such that its metric colision entropy for two cases: (a) seen by deterministic unbounded adversary and (b) seen by adversary using only randomized circuits of size only $\bigO{n}$ and accepting the distinguishing advantage to be only $1/\poly{n}$, differs by $\bigOm{\log \log n}$. Even more pathological result can be obtained for the Shannon Entropy: it is possible that the `gap' in the ammount of entropy for the unbounded deterministic and randomized adversary accepting constant distinguishing advantage, is even $\bigOm{n}$.

\section{Preeliminaries}

\paragraph{Information-theoretic notions}
The idea commonly used to define computational entropy is to generalize a convinient theoretic-infomation notion of entropy. Following this way, we start with recalling the notion of the R\'{e}nyi Entropy. 
\begin{definition}[R\'{e}nyi Entropy]\label{def:RenyiEntropy}
Given a random variable $X\in\{0,1\}^n$ we say that its R\'{e}nyi Entropy of order $\alpha$ (or in short: $\alpha$-Renyi Entropy) is at least $k$ if and only if
\begin{equation*}
 \left\|\mathbf{P}_{X}\right\|_{\alpha-1} = \left(\mathbf{E}_{x\leftarrow X} \left(\mathbf{P}_{X}(x)\right)^{\alpha-1}\right)^{\frac{1}{\alpha-1}} \leqslant 2^{-k}
\end{equation*}
\end{definition}
\noindent This definition covers also the important cases of the Colision Entropy ($\alpha=2$), the Shannon Entropy $(\alpha\to 1)$ and the Min-entropy $(\alpha\to\infty)$. By calculating these limits, one can give the explicity definitions for the last two cases:
\begin{definition}[Shannon Entropy]\label{def:ShannonEntropy}
Given a random variable $X\in\{0,1\}^n$ we say that its Shannon entropy is at least $k$ if and only if
\begin{equation*}
 -\sum\limits_{x}\PrD{X}(x) \log\PrD{X}(x) \geqslant k 
\end{equation*} 
where we define $p\log p = 0$ for $p = 0$.
\end{definition}
\begin{definition}[Min-Entropy]\label{def:MinEntropy}
Given a random variable $X\in\{0,1\}^n$ we say that its Min-Entropy is at least $k$ if and only if
\begin{equation*}
 \PrD{X}(x) \leqslant 2^{-k}\quad \text{for all } x\in\{0,1\}^n
\end{equation*} 
\end{definition}  
For some applications, like for the randomness extraction, the smoothed version of entropy is usefull. The key concept behind smooth entropy is that we allow $X$ to be only close (in some metric sufficiently strong to our purposes) to a distribution with required entropy, instead of expecting $X$ to satisfy the entropy constraints by itself. 
\begin{definition}[Statistical Disntance]
Let $X,Y \in \{0,1\}^n$ be two random variables. The statiscal distance of distributions $\PrD{X},\PrD{Y}$ is defined to be 
$\Delta(X,Y) = \frac{1}{2}\sum\limits_{x}\left| \PrD{X}(x) - \PrD{Y}(x)\right|$. 
\end{definition}
\begin{definition}[Smooth R\'{e}nyi Entropy]
Given $\epsilon>0$ and a random variable $X\in\{0,1\}^n$, we say that \emph{it has Smooth $\alpha$-R\'{e}nyi Entropy at least $k$}, if there exists a random variable $Y\in\{0,1\}^n$ such that $\mathbf{H}_{\alpha}\left(Y\right)\geqslant k$ and $\Delta(X,Y) \leqslant \epsilon$.
\end{definition}

\subsection{Computational Entropy}

The intuition behind HILL Entropy is that we think ofs $X$ as having high computational entropy if it is computationally indistinguishable from a distribution with (chosen) information-theoretic entropy. The computational variant of min-entropy was introduced in \cite{HILL99}. Below we generalize this concept replacing the min-entropy by R\'{e}nyi Entropy.
\begin{definition}[Computational HILL R\'{e}nyi Entropy]\label{HILL_Entropy_Definition}
Given $\epsilon > 0$, a class of disitinguishers\footnote{ The distinguishers can be deterministic or randomized $[0,1]$-valued functions.} $\mathcal{D}$ and a random variable $X\in\{0,1\}^{n}$, we say that \emph{ $X$ has at least $k$ bits of HILL Computational R\'{e}nyi Entropy of order $\alpha$ against $(\mathcal{D},\epsilon)$} and denote by $\mathbf{H}_{\alpha}^{\textup{HILL},\mathcal{D},\epsilon}(X)\geqslant k$ if there exist a distribution $Y$ over $\{0,1\}^{n}$ satisfying $\mathbf{H}_{\alpha}\left(Y\right)\geqslant k$ such that for any $D\in\mathcal{D}$ holds $\left|\mathbf{E}D(X) - \mathbf{E}D(Y)\right| \leqslant \epsilon$.
\end{definition}
\noindent Metric entropy is defined by reversing the order of quantifiers:
\begin{definition}[Computational Metric R\'{e}nyi Entropy]\label{Metric_Entropy_Definition}
Given $\epsilon > 0$ and a class of distnugishers $\cD$ we say that \emph{the random variable $X\in \{0,1\}^n$ has at least $k$ bits of Metric Computational R\'{e}nyi $\alpha$-Entropy against $(\mathcal{D},\epsilon)$} and denote by $\mathbf{H}_{\alpha}^{\textup{HILL},\mathcal{D},\epsilon}(X)\geqslant k$ if for any $D\in \cD$ there exist a distribution $Y$ over $\{0,1\}^{n}$ satisfying $\mathbf{H}_{\alpha}\left(Y\right)\geqslant k$ and $\left|\mathbf{E}D(X) - \mathbf{E}D(Y)\right| \leqslant \epsilon$.
\end{definition}

\subsection{Conditional Computational Entropy}

The conditional computational entropy is defined in the similar way via underlying theoretic-information entropy measure. Since there is no agreement how to define Conditional R\'{e}nyi Entropy, to make this discussion clear we restrict us only to the case of min-entropy. Usually one defines the conditional min-entropy in one of the two ways:
\begin{definition}[Conditional Min Entropy]
Given a joint distribution $(X,Z)$ we say that \emph{$X$ conditioned on $Z$ has min-entropy at least $k$} and denote by $\mathbf{H}(X|Z)\geqslant k$ if
\begin{equation*}
 \forall z:\, \mathbf{H}_{\infty}\left(X|Z=z \right) \geqslant k
\end{equation*} 
\end{definition}
\begin{definition}[Average Conditional Min Entropy \cite{Dodis2008}]
Given a joint distribution $(X,Z)$ we say that \emph{ $X$ conditioned on $Z$ has average min-entropy at least $k$} and denote by $\widetilde{\mathbf{H}}(X|Z)\geqslant k$ if
\begin{equation*}
 \mathbf{E}_{z\leftarrow Z}\left[ 2^{-\mathbf{H}\left(X|Z=z\right)} \right] = \mathbf{E}_{z\leftarrow Z}\left[ \max\limits_{x}\mathbf{P}_{X|Z=z}(x) \right] \leqslant 2^{-k}
\end{equation*}  
\end{definition}
\noindent The conditional computational entropy is defined similarly to the unconditional case. 
\begin{definition}[Conditional Computational HILL Entropy]\label{HILL_Conditional_Entropy_Definition}
Given $\epsilon > 0$, a class of distinguishers $\cD$ and a pair of random variables $X\in \{0,1\}^{n},Z\in\{0,1\}^m$ we say that \emph{$X$ conditioned on $Z$ has at least $k$ bits of HILL Computational Min Entropy against $(\mathcal{D},\epsilon)$} and denote by $\mathbf{H}_{\infty}^{\textup{Metric},\mathcal{D},\epsilon}(X|Z)\geqslant k$ if there exists a distribution $(Y,Z)$ over $\{0,1\}^{n+m}$ satisfying $\mathbf{H}_{\infty}\left(Y|Z\right)\geqslant k$ such that for any $D\in\mathcal{D}$ holds the inequality
$\left|\mathbf{E}D(X,Z) - \mathbf{E}D(X,Z)\right| \leqslant \epsilon$.
\end{definition}
The conditional computational metric entropy is defined by changing the order of the quantifiers. Metric as well as HILL conditional entropy can be defined as average or non-average conditional entropy depending on use $\widetilde{\mathbf{H}}_{\infty}$ or $\mathbf{H}_{\infty}$ and denoted using these symbols. For clarity we do not give the rest of possible definitions.

For the sake of completeness we note that there is a definition that allows $Z$ to change together with $Y$. This leads to the notion of relaxed computational min entropy:
\begin{definition}[Conditional Computational HILL Relaxed Entropy, \cite{Reyzin2011}]\label{HILL_Computational_Relaxed_Entropy_Definition}
Given $\epsilon > 0$, a class of distinguishers $\cD$ and a pair of random variables $X\in \{0,1\}^{n},Z\in\{0,1\}^m$, se say that \emph{$X$ conditioned on $Z$ has at least $k$ bits of Relaxed HILL Computational Entropy against $(\mathcal{D},\epsilon)$} and denote by $\mathbf{H}_{}^{\textup{HILL-rlx},\mathcal{D},\epsilon}(X|Z)\geqslant k$ if there exists a distribution $(Y,Z')$ over $\{0,1\}^{n+m}$ satisfying $\mathbf{H}_{\infty}\left(Y|Z'\right)\geqslant k$ such that for any $D\in\mathcal{D}$ holds
$\left|\mathbf{E}D(X,Z) - \mathbf{E}D(X,Z')\right| \leqslant \epsilon$.
\end{definition} 
This entropy also can be considered in average or non-average aspects, with HILL or Metric type of indistingusihability and has remarkably good properties for some leakage-related problems, as we will see later.

\paragraph{Relationships between HILL and Metric Entropy}
The Metric entropy, which was was introduced after the HILL one, often turns out to be more convenient in applications (for instance, to prove leakage-related results). It is known that from Metric Entropy computed against real valued (or randomized) circuits, then there exists a conversion to HILL entropy \cite{Barak2003}. This result in its full generality can be stated as follows
\begin{theorem}[Generalization of \cite{Barak2003}, Thm. 5.2]\label{MetricHILL_Conversion}
Let $\cP$ be the set of all probability measures over $\Omega$. Suppose that we are given a class $\mathcal{D}$ of $[0,1]$-valued functions on $\Omega$, with the following property: if $D\in \cD$ then $D^{c}=^{\textup{def}}\mathbf{1}-D \in \cD$. For $\delta > 0$, let $\mathcal{D'}$ be the class consisting of all convex combinations of length $\mathcal{O}\left(\frac{\log|\Omega|}{\delta^2}\right)$ over $\mathcal{D}$. Let $\mathcal{C}\subset \cP$ be any arbitrary convex subset of probability measures and $X \in \cP$ be a fixed distribution. Consider the following statements:
\begin{enumerate}[label = \roman{*}]
\item \label{itm:HILLMetricCoversion_a} $X$ is $\left(\mathcal{D},\epsilon+\delta\right)$ indistinguishable from \emph{some} distribution $Y\in \mathcal{C}$ (HILL Entropy)
 \item \label{itm:HILLMetricCoversion_b} $X$ is $\left(\mathcal{D'},\epsilon\right)$ indistinguishable from the set of \emph{all} distribution $Y\in \mathcal{C}$ (Metric Entropy)
\end{enumerate}
Then (\ref{itm:HILLMetricCoversion_b}) implies (\ref{itm:HILLMetricCoversion_a}).  
\end{theorem}
The sketch of the proof appears in Appendix.
\begin{remark}\label{MetricHILL_Conversion_Discussion}
By choosing $\Omega = \{0,1\}^{n+m}$, a random variable $Z\in\{0,1\}^{m}$ and $\mathcal{C}$ to be the set of all distributions $(Y,Z)$ satisfying $(Y,Z):\,\Hmin{Y|Z}\geqslant k$ or alternatively 
$\mathbf{\widetilde{H}}_{\infty}(Y|Z)\geqslant k$, we obtain the conversion from the Metric Conditional Entropy to the HILL Conditional Entropy, for both: worst case and average case variants.
\end{remark}

\subsection{Relationship of Convex Analysis to Metric Entropy}\label{Relationship_to_ConvexAnalysis} 

Let us notice, that the both notions: HILL and Metric entropy can be rephrased in a geometrical language as unability to separate between two convex sets. More precisely, any distribution $\PrD{X}$ on $\{0,1\}^n$, after enumerating the elements of $\{0,1\}^n$, can be (uniquely) identified with a vector in $\mathbb{R}^{2^n}$. Similarly, any real valued function $D$ on $\{0,1\}^{n}$ can be identified with a vector in the same space. Taking the expected value becomes then the scalar product 
\begin{equation*}
 \langle D, \PrD{X} \rangle = \sum\limits_{x}\PrD{X}(x) D(x) = \mathbf{E}_{x\leftarrow X} D(x)
\end{equation*}
By considering the min-entropy, for instance, it is easy so see that $\mathbf{H}_{\infty}^{\text{Metric},\cD,\epsilon}(X) < k$ if and only if there exists $D\in\cD$ such that for all $\PrD{Y}\in\cY$
\begin{align}
 \left| \langle D, \PrD{X} - \PrD{Y} \rangle \right|  = & \left| \sum\limits_{x}D(x)\left(\PrD{X}(x) - \PrD{Y}(x)\right)\right|  
  \geqslant \left| \mathbf{E}D(X) - \mathbf{E}D(Y) \right| \geqslant \epsilon \nonumber
\end{align}
where $\cY$ is the set of all distributions on $\{0,1\}^{n}$s with min-entropy at least $k$. We will see later that the absolute value above can be removed by considering classes $\cD$ which are closed under complements (i.e. if $D\in \cD$ then also $\mathbf{1}-D \in\cD$). Then we get the
inequality $\langle D, \PrD{X} \rangle \geqslant \langle D, \PrD{Y} \rangle + \epsilon$ valid for all $\PrD{Y}\in\cY$. Thus defining the metric entropy is nothing more than just saying that a given distribution $X$ \emph{cannot be separated} (in the sense known from functional analysis or convex analysis) from the \emph{set} $\cY$ (i.e. from all its elements at once). In the other hand, $D$ that contraddicts the definition is exactly \emph{a separating hyperplane}. Hence, methods of convex analysis can be applied to study the properties of  metric-type entropies. The HILL-type definition is less compatible with this approach, as it is a bit stronger assumption, namely that we are not able to separate any pair 
$\PrD{X},\PrD{Y}$ where $\PrD{Y} \in \cY$. In this paper we follow the terminology used in computer science, saying about \emph{distinguishing} instead of \emph{separating} as in math.

\subsection{Used conventions and important remarks}
Through this paper we will use mostly the already defined computational min-entropy, saying in short about computational entropy. We will thereby often omit the sign $\infty$ writing $\mathbf{H}^{\textup{HILL},\mathcal{D},\epsilon}$,  $\widetilde{\mathbf{H}}^{\textup{Metric-rlx},\mathcal{D},\epsilon}$ and so on when meaning min-entropy based computational entropy. We also use the following natural convention: replacing $\mathcal{D}$ by a pair $(\{0,1\},s)$ or $([0,1],s)$ if we mean deterministic circuits of size $s$ respectively boolean and $[0,1]$-valued. Writing $(\textup{rand}\{0,1\},s)$ in the place of $\mathcal{D}$ we mean randomized boolean circuits of size $s$. If the circuit size $s$ is omitted in the description of a circuit class, it is assumed to be unbounded. For the boolean function $D$ we denote $|D| = \sum_{x\in \textup{dom} D}D(x)$.

Note that although one can define and use computational entropy based on the  R\'{e}nyi Entropy of any arbitrary order, using of min-entropy as a reasonable compromise between the the convenience of analysis and preserving so much generality as possible, is not a big restriction in practice, as long as one uses 
real valued distinguishers. To pass between Renyi Entropies for different order, one uses the fact that the values of the Smooth R\'{e}nyi Entropy for different order cannot be differ more than a small additive constant. The precise statement is given bellow: 
\begin{lemma}[\cite{Renner2004}]\label{RenyiEntropy_DifferentOrders}
Suppose that $X$ is a distribution over $\{0,1\}^{n}$.  Then for $\alpha>1$
\begin{equation*}
  \mathbf{H}_{\infty}^{\epsilon}(X) \geqslant \mathbf{H}_{\alpha}(X)-\frac{1}{\alpha-1}\log\frac{1}{\epsilon}
\end{equation*}
\end{lemma}
\noindent We will obtain another proof of this result using a characterization of computational metric entropy.
These equivalence does not cover the Shannon Entropy case. It is worth of noting that the Shannon Entropy based Computational Entropy also found applications where it becomes more suitable than the computational min-entropy \cite{Vadhan2012}.

\section{Separation between Computational Entropy and Smooth Entropy}
In this section we examine the existence of a conversion rule from the computational to 
the smooth entropy:
\begin{quote} 
Suppose that $\Hmtr{X|Z}{\det [0,1], s, \epsilon} \geqslant k$, where $X\in\{0,1\}^{n},Z\in\{0,1\}^{m}$. What are the conditions on $s,\epsilon$ that guarantee that
$\mathbf{H}_{\infty}^{\epsilon'}(X|Z) \geqslant k'$ with $\epsilon'\leqslant 2^{C}\epsilon$ and $k'\geqslant k-C$ for some constant $C$?
\end{quote}
In Section \ref{sec:5} we will prove that if the security parameters are sufficiently strong, more precisely, if $s = \bigO{2^{k+m}}$, then the computational min-entropy becomes the smooth entropy. For the unconditional case also
exponentially small $\epsilon$ is sufficient (see Section \ref{sec:5}, Corollary \ref{MinimalComplexityRequired}).
Interestingly, this result can be inverted. In this section we show that
\begin{quote}
Given the metric entropy of $X|Z$ one really needs the security to be exponentially strong in $k+m$, to obtain the smooth entropy with comparable parameters.
\end{quote}
We stress that although the existence of a separation between the Metric and Smooth Entropy is almost obvious, the quantitative bound which is exponential in both: $k$ and $m$ is less triviall to see. Since the maximal entropy of $X|Z$ is $n$, it follows that even if distinguishers were given access to an oracle over $\{0,1\}^{n}$, the entropy could be still non-trivial.

\begin{remark}
Since Min-Entropy is the smallest one among other Renyi Entropies, and because there is efficient conversion for Smooth Renyi Entropies (Lemma \ref{RenyiEntropy_DifferentOrders}) it is sufficient to consider the case of Min Entropy.
\end{remark}

\paragraph{Separation for unconditional Computational Min-entropy} 

\begin{theorem}\label{Separation_MetricSmooth_Unconditional}
For any $C>1$, there exists $X$ such that $\Hmtr{X}{\bigOm{\nicefrac{2^{k}\epsilon^2}{\log\left(2^{k}\epsilon^2\right)}},\epsilon} \geqslant k+C$ but $\Hsmooth{X}{\nicefrac{1}{2}}\leqslant k+1$. 
\end{theorem}
\begin{proof}
The main idea is to reduce the problem to a problem of \emph{approximating of a certain function}, which will turn out to be \emph{hard} for limited size circuits. 
Let $A$ and $S$ be sets of cardinality $2^k$ and $2^{k+C}$ and $A\subset S$. Denote $B = S\setminus A$. Consider the random variable $X=U_{A}$. It is easy to see that $\Hmtr{X}{\frac{1}{2}} \leqslant k+1$. Observe that
\begin{align*}
 \mathbf{E}D(X) - \mathbf{E}D\left(U_{S}\right) & =  \\
 = &  \mathbf{E}D\left(U_{A}\right) - \mathbf{E}D\left(U_{S}\right)  \\ 
 = & \left(1-2^{-C}\right) \left( \Pr{}{D\left(U_{A}\right)=1} -  \Pr{}{D\left(U_{B}\right)=1}\right)  \\
 = & \left(1-2^{-C}\right) \left( \Pr{}{D\left(U_{A}\right)=1} +  \Pr{}{D\left(U_{B}\right)=0}\right) - \left(1-2^{-C}\right),
\end{align*}
hence, assuming $\Hmtr{X}{\textup{det}\{0,1\},s,\epsilon} < k+C$ for $\epsilon = \delta \left(1-2^{-C}\right)$, we get 
\begin{equation*}
 \Pr{}{D\left(U_{A}\right)=1} +  \Pr{}{D\left(U_{B}\right)=0} > 1+\delta
\end{equation*}
The proof easily follows now from the following lemma, being a strenghtening of the classical result on the existence of $\delta$-hard functions.
\begin{lemma}\label{hard_functions_1}
For any $C \geqslant 1$ and sufficiently large $\ell$ there exists a boolean function $f$ over $\{0,1\}^{\ell}$, such that $\mathrm{bias}(f) = 1-2^{1-C}$ and for all circuits $D$ of size $\bigO{\nicefrac{2^{\ell-C}\delta^2}{\ell-C-2\log(1/\delta}}$ we have
\begin{equation*}
 \Pr{x\leftarrow f^{-1}(\{1\})}{D(x) = f(x)} + \Pr{x\leftarrow f^{-1}(\{0\})}{D(x) = f(x)} < 1+\delta
\end{equation*}
\end{lemma}
The proof follows by a standard application of the Chernoff Bound and the union bound over the all circuits of bounded size. See Appendix, \ref{app:hard_functions_1} for the details and discussion. 
\end{proof}

\paragraph{Separation for Conditional Computational Min Entropy}

\begin{theorem}\label{Separation_MetricSmooth_Conditional}
For sufficiently large $n$, and for any $C > 0$, $k<n-C$ and $\epsilon > 0$ there exists a pair of jointly distributed random variables $X\in \{0,1\}^n$, $Z\in\{0,1\}^m$ such that 
\begin{enumerate}[label = (\roman{*})]
 \item $\mathbf{H}^{1/2}_{\infty}(X|Z) \leqslant k+1$
 \item $\mathbf{H}^{\text{Metric},\det  [0,1], s, \epsilon}(X|Z) \geqslant k+C$ for $ s =  \bigOm{\frac{2^{k+m}\epsilon^4}{(k+m)\log \left(2^{k+m}\epsilon^2 \right)}}$
\end{enumerate}
\end{theorem}
The proof is longer than for the unconditional case. The key point is that in the conditional case it is significantly harder to find an appropriate ''hard approximation task". See the proof of Theorem \ref{app:Separation_MetricSmooth_Conditional} in the Appendix.

\section{Characterizations of Metric Entropy}

\paragraph{General Characterization Theorem}

\noindent
The following result can be viewed as a general characterization of Metric Entropy.  The easy proof is given in the Appendix.
\begin{theorem}\label{MetricEntropy_Characterization}
Let $\mathcal{D}$ be a class of real valued functions on $\{0,1\}^n$ closed under complements, $\cY$ be a non-empty compact convex set of probability distributions over $\{0,1\}^n$ and $X\in\{0,1\}^n$ be a random variable. Then the following conditions are equivalent:
\begin{enumerate}[label = (\roman{*})]
 \item For every $Y\in\cY$ there exists $D\in \cD$ such that $\left|\mathbf{E}D(X) - \mathbf{E}D(Y)\right| < \epsilon$
 \item For every $D\in \cD$ we have $\mathbf{E}D(X) \leqslant \max\limits_{Y\in \cY} \mathbf{E}D(Y) + \epsilon$
\end{enumerate}
\end{theorem}
\begin{remark}
If the entropy is defined based on underlying information-theoretic entropy measure $\mathbf{H}$, for instance Renyi Computational Entropy defined in Section 2, then the set $\cY$ is just so called superlevels set: it consists of all distributions having the (information-theoretic) entropy at least $k$. For the conditional relaxed entropy of $X|Z$, we set $\cY = \left\{\PrD{Y,Z'}:\, Y\in \{0,1\}^n, Z'\in\{0,1\}^m,\mathbf{H}_{\infty}(X|Z')\geqslant k \right\}$.
\end{remark}

\noindent It is clear that we need to solve the maximization task explicity, in order to obtain a characterization for a concrete variant of Metric Entropy. 

\paragraph{Renyi Entropy} \

\noindent By computing $\max_{Y\in \cY} \mathbf{E}D(Y) $ in Theorem \ref{MetricEntropy_Characterization}, we characterize the most important cases of the Metric R\'{e}nyi Entropy.
\begin{lemma}\label{RenyiEntropy_ExtremeDistributions}
Let $\alpha>1$ be fixed, let $D:\{0,1\}^{n}\rightarrow \{0,1\}$ be a function and $\mathcal{Y}_k = \left\{Y\in\{0,1\}^{n}:\, \mathbf{H}_{\alpha}(Y)\geqslant k \right\}$. Then 
\begin{equation}
 \max\limits_{Y\in \mathcal{Y}_k} \mathbf{E}D(Y) = 
\left\{
\begin{array}{rl}
 p_D \cdot |D|,  & \textup{if } |D| < 2^k \\
 1, & \textup{otherwise}
\end{array}
\right.
\end{equation}
where $p_D$, for $|D| \leqslant 2^k$, is the greatest number satisfying the following system 
\begin{equation}\label{eq:RenyiEntropy_ExtremeDistributions}
\left\{
\begin{array}{rcl}
 p_D^{\alpha}|D| + q_D^{\alpha}|D^c| &=& 2^{-(\alpha-1)k}  \\
 p_D |D| + q_D |D^c| &=&  1 \\
 p_D,q_D & \geqslant & 0
\end{array}\right. 
\end{equation}
Moreover, the solution $p_D$ is unique provided that $k < n-1$.
\end{lemma}
The proof is not hard but technical and is left to the appendix (see the proof of Theorem \ref{app:RenyiEntropy_ExtremeDistributions}). Especially, for the cases $\alpha \to 1$, $\alpha=2$ and $\alpha\to\infty$ corresponding to the Shannon, Colision and Min-Entropy respectively, after some calculations we obtain the following characterizations:
\begin{corollary}[Metric Shannon Entropy]\label{MetricShannonEntropy_Characterization}
Let $\mathcal{D}$ be a class of boolean functions closed under complements. Then the following conditions are equivalent
\begin{enumerate}[label = (\roman{*}))]
  \item $\mathbf{H}_{1}^{\textup{Metric},\mathcal{D},\epsilon}(X)\geqslant k$ 
  \item For every $D\in\mathcal{D}$ such that $|D|\leqslant 2^{k}$,
 $\mathbf{E}D(X) \leqslant p_D |D| + \epsilon$ holds for $p_D$ solving the system
\begin{equation}\label{eq:ShannonEntropy_ExtremeDistributions}
\left\{
\begin{array}{rcl}
 -p_D |D| \log p_D  - q_D \left|D^c\right| \log q_D  &=& k  \\
 p_D |D| + q_D |D^c| &=&  1 \\
 p_D,q_D &\geqslant & 0
\end{array}\right. 
\end{equation} 
\end{enumerate} 
\end{corollary}

\begin{corollary}[Metric Colision Entropy]\label{MetricColisionEntropy_Characterization}
Let $\mathcal{D}$ be a class of boolean functions closed under complements. Then the following conditions are equivalent
\begin{enumerate}[label = (\roman{*})]
  \item $\mathbf{H}_{2}^{\textup{Metric},\mathcal{D},\epsilon}(X)\geqslant k$ 
  \item The inequality $\mathbf{E}D(X) \leqslant p_D |D| + \epsilon$ holds for every $D\in\mathcal{D}$, and $p_D$ given by
\begin{equation}\label{eq:ColisionEntropy_ExtremeDistributions}
p_D = 2^{-n}+\sqrt{\left|D^c\right| |D|^{-1}\left(2^{-k-n}-2^{-2n}\right)}
\end{equation} 
\end{enumerate} 
\end{corollary}

\begin{corollary}[Metric Min-Entropy, \cite{Barak2003}]\label{MetricMinEntropy_Characterization}
Let $\mathcal{D}$ be a class of boolean functions closed under complements. Then the following conditions are equivalent
\begin{enumerate}[label = (\roman{*})]
  \item $\mathbf{H}_{\infty}^{\textup{Metric},\mathcal{D},\epsilon}(X)\geqslant k$ 
  \item The inequality $\mathbf{E}D(X) \leqslant 2^{-k} |D| + \epsilon$ holds for every $D\in\mathcal{D}$
\end{enumerate} 
\end{corollary}
\noindent Note that from the characterization in Lemma \ref{MetricEntropy_Characterization}, and Lemma \label{RenyiEntropy_ExtremeDistributions} it follows that one need to check only distinguishers of size $\exp(k)$ to prove that the metric entropy is $k$.
\begin{corollary}\label{MinimalComplexityRequired}
Let $X\in\{0,1\}^n$ be a random variable, $\alpha\in [1,\infty]$ and $s = \bigOm{2^{k}k}$. Then $\mathbf{H}_{\alpha}^{\text{Metric},\{0,1\},s} = \mathbf{H}_{\alpha}^{\text{Metric},\{0,1\},\infty}$. 
\end{corollary}

\paragraph{Relaxed Computational Entropy}\

\begin{lemma}
Let $X \in \{0,1\}^n,Z \in \{0,1\}^m$ be random variables, and let 
\begin{equation}
 \cY = \left\{\PrD{Y,Z'}:\, Y\in\{0,1\}^n,Z'\in\{0,1\}^m, \Hmin{Y|Z'} \geqslant k \right\}
\end{equation}
Then for every boolean function $D$ on $\{0,1\}^{n+m}$ we have
\begin{equation}
 \max\limits_{\PrD{Y,Z'}\in \cY} \mathbf{E}D(Y,Z')=  2^{-k} \min\left( \max\limits_{z} \left|D(\cdot,z)\right|, 1\right)
\end{equation}
\end{lemma}
\begin{proof}
Observe that for every $\PrD{Y,Z'}\in \cY$ we have
\begin{align}
 \mathbf{E}D(Y,Z') & \leqslant \sum\limits_{x,z} \PrD{X,Z}(x,z) D(x,z) \\
 & \leqslant \sum\limits_{x,z} 2^{-k}\PrD{Z}(z) D(x,z) \\
 & \leqslant 2^{-k}\max\limits_{z}\left| D(\cdot,z) \right|
\end{align}
Let $z'$ be chosen over $z\in\{0,1\}^m$ so that it maximizies $\left| D(\cdot,z) \right|$.
The equality in the estimate above is achieved provided that $\left| D(\cdot,z') \right|\leqslant 2^{k}$,
$Z'$ is a point mass distribution at $z'$ and $Y|Z'=z'$ satisfies $\PrD{Y|Z'=z'}(x) = 2^{-k}$ if $D(x,z)=1$.
For the case $\max\limits_{z}\left| D(\cdot,z) \right| > 2^{k}$, let $Y'|Z'=z'$ be a uniform distribution on arbitrary $2^k$-element subset of $\{x:\, D(x,z) =1 \}$. Then we have $\mathbf{E}D(Y,Z') = 1$.
\end{proof}
\begin{corollary}[Relaxed Metric Entropy] \label{MetricRelaxedEntropy_Characterization}
Let $\mathcal{D}$ be a class of boolean functions closed under complements. Then the following are equivalent
\begin{enumerate}[label = (\roman{*})]
  \item $\mathbf{H}_{\infty}^{\textup{Metric-rlx},\mathcal{D},\epsilon}(X|Z)\geqslant k$ 
  \item The inequality $\mathbf{E}D(X,Z) \leqslant 2^{-k} \max\limits_{z} \left|D(\cdot,z)\right|  + \epsilon$ holds for every $D\in\mathcal{D}$
\end{enumerate}  
\end{corollary}

\subsection{Examples of Applications}

As a first example we give below much simpler proofs of the leakage chain rule for relaxed-type entropy and for the leakage lemma. Interestingly, there is no hope for proving an efficient version (meaning a bound on loss in security parameters) for non-relaxed definition, as shown recently by Krenn et al.\ \cite{Pietrzak2012}.
\begin{theorem}[Leakage Lemmas]
Let $X,Z_1,Z_2$ be correlated random variables taking values in $\{0,1\}^n,\{0,1\}^m_1$ and $\{0,1\}^m_2$ respectively. Then the following estimate, called ``the chain rule", is true \cite{GentryWichs2010, Reyzin2011}
\begin{equation*}
 \mathbf{H}^{\textup{Metric-rlx},\{0,1\},s,2^{m_2}\epsilon}_{\infty}\left(X\left|Z_1,Z_2\right.\right)\geqslant \mathbf{H}^{\textup{Metric-rlx},\{0,1\},s,\epsilon}_{\infty}\left(X\left|Z_1\right.\right)-m_2.
\end{equation*}
Especially, for the uncoditional case $Z_1 = \emptyset$) we obtain the so-called ``leakage lemma" \cite{Dziembowski2008, Reingold2008, Fuller2011}
\begin{equation*}
 \widetilde{\mathbf{H}}^{\textup{Metric},\{0,1\},s,2^{m_2}\epsilon}_{\infty}\left(X\left|Z_2\right.\right)\geqslant \mathbf{H}^{\textup{Metric},\{0,1\},s,\epsilon}_{\infty}\left(X \right)-m_2.
\end{equation*}
\end{theorem}
\begin{proof}
Let $D$ be any boolean distinguisher on $\{0,1\}^{n+m_1+m_2}$. Since for every fixed $z_2$ the function $D\left(\cdot,z_2\right)$ is a distinugisher on $\{0,1\}^{n+m_1}$ we get from Corollary \ref{MetricRelaxedEntropy_Characterization}
\begin{equation}
 \mathbf{E}_{\left(x,z_1\right)\leftarrow \left(X,Z_1\right)} D\left(x,z_1,z_2\right) \leqslant \max\limits_{z_1}\left| D\left(\cdot,z_2\right)\right|\cdot 2^{-k} + \epsilon
\end{equation}
note also that $\mathbf{E}_{\left(x,z_1\right)\leftarrow \left(X,Z_1\right)|Z_2=z_2} D\left(x,z_1,z_2\right) \leqslant \frac{1}{\mathbf{P}_{Z_2}\left(z_2\right)}\cdot\mathbf{E}_{\left(x,z_1\right)\leftarrow \left(X,Z_1\right)}D\left(x,z_1,z_2\right)$ and thus
\begin{align}
 \mathbf{E}_{\left(x,z_1,z_2\right)\leftarrow \left(X,Z_1,Z_2\right)}D(x,z) = &
\mathbf{E}_{z_2\leftarrow Z_2} \mathbf{E}_{\left(x,z_1\right)\leftarrow \left(X,Z_1\right)|Z_2=z_2} D\left(x,z_1,z_2\right) \\ \leqslant & \sum\limits_{z_2}\mathbf{E}_{\left(x,z_1\right)\leftarrow \left(X,Z_1\right)}D\left(x,z_1,z_2\right)  \\
\leqslant & 2^{m_2}\max\limits_{z_2}\mathbf{E}_{\left(x,z_1\right)\leftarrow \left(X,Z_1\right)} D\left(x,z_1,z_2\right) \\
\leqslant & \max\limits_{z_1,z_2}\left| D\left(\cdot,z_1,z_2\right)\right| \cdot 2^{-k+m_2} + 2^{m_2}\epsilon
\end{align}
using Lemma \ref{MetricRelaxedEntropy_Characterization} again finishes the proof.
\end{proof}
\begin{remark}
Note that both results are often formulated using the HILL entropy, with the wekaer, by a factor $\poly{1/\epsilon}$, security parameter $s$. This factor is exactly the cost of the conversion from Metric to HILL entropy (Theorem \ref{MetricHILL_Conversion} can be applied for the relaxed-metric entropy, because for this notion boolean and real valued circuits are equivalent as will see later; the same is true for the conditional metric worst-case entropy which is known to be equivalent to the conditional metric average-case entropy up to loss $\log(1/\epsilon)$ in the entropy ammount). Sometimes loss appears in $\epsilon$ instead in $s$, which can be also thought as an equivalent statement. Thus, our proof really \emph{implies} the original results.
\end{remark}
In the second example, we reprove a result on passing between Renyi Entropy for different orders.
\begin{corollary}
The proof of Lemma \ref{RenyiEntropy_DifferentOrders}.
\end{corollary}
\begin{proof}
Since $p$ is a solution of the equation $ p^{\alpha}|D| + \left( \frac{1-p|D|}{\left|D^c\right|}\right)^{\alpha} = 2^{-(\alpha-1)k}$ we have $p|D| < \left(\frac{|D|}{2^k}\right)^{\frac{\alpha-1}{\alpha}}$. Suppose that $|D| > 2^{k}\cdot 2^{-C\alpha}$ where the parameter $C$ will be specified later. Then
$\left(\frac{|D|}{2^k}\right)^{\frac{\alpha-1}{\alpha}} < \frac{|D|}{2^{k-C}}$. In turn if the opposite inequality holds then $\left(\frac{|D|}{2^k}\right)^{\frac{\alpha-1}{\alpha}} \leqslant 2^{-C(\alpha-1)}$. In any case, we get
$p|D| \leqslant \frac{|D|}{2^{k'}}+\epsilon'$ where $k'=k-C$ and $\epsilon' = 2^{-C(\alpha-1)}$. Thus the distribution $X$
has a metric entropy at least $k'$ with error $\epsilon'$ against all bollean functions. According to Theorem \ref{MetricRelaxed_passing_to_realvalued} in the next section this is the same entropy as if $[0,1]$ valued functions would be used. From Theorem \ref{MetricHILL_Conversion} giving conversion between Metric and HILL entropy we know that $X$ has HILL Entropy with the same parameters. Finally, HILL Entropy against all $[0,1]$ valued function is clearly the same as Smooth Entropy. Choosing $C = \frac{1}{\alpha-1}\log\frac{1}{\epsilon}$ we recover the estimate on the Smooth R\'{e}nyi Entropy for different orders given in Lemma \ref{RenyiEntropy_DifferentOrders}.
\end{proof}

\section{Metric Entropy Against Deterministic and Randomized Adversary}\label{sec:5}

It is well known that for the HILL-type entropy there is no matter whether we use deterministic or randomized (or real valued) class of distinguishers. The reason is that we can just fix an `optimal' choice of coins for a randomized function distinguishing between two probability distributions. However, this argument fails in the case of a metric-type definition because of a different order of quantifiers in the definition. So the following problem appears:
\begin{quote}
Let $\mathbf{H}^{\textup{Metric}}$ be metric-type computational entropy (for instance, 
based on the Renyi Entropy of fixed order $\alpha$). Suppose that $\mathbf{H}^{\textup{Metric},\mathrm{det} \{0,1\},s,\epsilon}(X) \geqslant k$. Can we obtain a good lower bound on 
$\mathbf{H}^{\textup{Metric},\mathrm{det}, [0,1],s',\epsilon'}(X)$ or on $\mathbf{H}^{\textup{Metric},\mathrm{rand}, \{0,1\},s',\epsilon'}(X)$ ?
\end{quote}

\subsection{Positive answer for min-entropy}

We show that even in the conditional case, for the min-entropy based metric and relaxed-metric entropy, boolean and real-valued distinguishers are equivalent. 
\begin{theorem}\label{MetricWorstCase_passing_to_realvalued}
Let $X\in\{0,1\}^n$ and $Z\in\{0,1\}^n$ be random variables. Then we have $\Hmtr{X|Z}{\textup{det},\{0,1\},s,\epsilon} = \Hmtr{X|Z}{\textup{det},[0,1],s',\epsilon}$ where $s'\approx s$. 
\end{theorem}
\begin{theorem}\label{MetricRelaxed_passing_to_realvalued}
Let $X\in\{0,1\}^n$ and $Z\in\{0,1\}^n$ be random variables. Then we have $\mathbf{H}_{\infty}^{\textup{Metric-rlx}, \det[0,1], s,\epsilon}(X|Z)= \mathbf{H}_{\infty}^{\textup{Metric-rlx}, \det[0,1], s',\epsilon}(X|Z)$ where $s'\approx s$. 
\end{theorem}
The idea of the proof is to rephrase the problem as a task of separating convex sets, as discussed in Section \ref{Relationship_to_ConvexAnalysis}. The standard proof by reduction requires to construct a boolean distinguisher from the given possibly real-valued one. In terms of convex analysis, it becomes a task of finding an appropriate (satysfying some restrictions) separating hyperlane. Technically this is done by calculating \emph{Lagrange Multipliers}.
The details are given in the Appendix, Theorem \ref{app:MetricWorstCase_passing_to_realvalued} and\ref{app:MetricRelaxed_passing_to_realvalued}. Passing further to randomized circuits can be realized (with a loss) using Theorem \ref{MetricHILL_Conversion}.

\subsection{Negative results for Renyi Entropy of order $\alpha < \infty$}

Having shown some positive results for the min-entropy based metric entropy, we will show a surprising property: for any other R\'{e}nyi Entropy, there exists a random variable such that its entropy against deterministic boolean circuits is \emph{strictly smaller} that the entropy against real-valued circuits (and therefore also randomized circuits). Before we show the actual proof, let us give some geometric intuitions why is the min-entropy so special. The reason is, that the set of all distributions having min-entropy at least $k$, after encoding probabilities as vectors, is given by linear inequalities of the form $0 \leqslant p_i \leqslant 2^{-k}$ and $\sum\limits_{i} p_i = 1$. Since that all inequality constraints form a hypercube whose faces are given by $0-1$ vectors, they ''match" perfectly to the boolean distinguishers very well. Compare this to the colision entropy, where the entropy (collision) constraint is $\sum\limits_{i} p_i^2 \leqslant 2^{-k}$ and the corresponding shape is clearly an ellipsoid.

From the characterization given in Lemma \ref{RenyiEntropy_ExtremeDistributions} we inmediatelly obtain
\begin{proposition}\label{MetricEntropy_BooleanFooled}
Let $\mathcal{Y}_k$ be the set of all distributions over $\{0,1\}^n$ with the $\alpha$-Entropy at least $k$.
Then the set of distributions $X$ over $\{0,1\}^n$ which are $\epsilon$-nonindistinugishable from $\mathcal{Y}_k$  by boolean functions (i.e. $\mathbf{H}^{\textup{Metric},\mathrm{det}\{0,1\},\epsilon}(X) \geqslant k$) 
is described by the following system of inequalities
\begin{equation}
 \PrD{X}:\, \langle D, \PrD{Y} \rangle \leqslant p_D |D| + \epsilon \quad D \in \cD,\  \textup{} p_D \textup{ satisfies } (\ref{eq:RenyiEntropy_ExtremeDistributions})
\end{equation}
\end{proposition}

\begin{corollary}
For $1\leqslant \alpha < \infty$ there exist a random variable $X\in\{0,1\}^n$ and $\epsilon_n > 0$ such that 
\begin{equation*}
 \mathbf{H}_{\alpha}^{\textup{Metric},\mathrm{det} \{0,1\},0}(X) > \mathbf{H}_{\alpha}^{\textup{Metric},\mathrm{rand},\{0,1\},\epsilon_n}(X) 
\end{equation*}
\end{corollary}
\begin{proof}
Note that since the class $\mathrm{det}\{0,1\}$ of all boolean functions on $\{0,1\}^n$ is finite
and since for every $D$ there is only finitely many solutions $p_D$, the set of solutions of (\ref{eq:RenyiEntropy_ExtremeDistributions}) is a convex polyhedron on 
the simplex of all probability measures over $\{0,1\}^n$ thought as a subset of the space $\mathbb{R}^{2^n}$.
On the other hand, the set $\cY$ of all distributions $Y$ with entropy $\mathbf{H}_{\alpha}(Y)$ at least $k$, being its subset, cannot be a polyhedron as it is defined by the smooth function
($v\rightarrow \sum\limits_{i=1}^{2^n} v_i^{\alpha}$ if $1 < \alpha < \infty$ and $v\rightarrow \sum\limits_{i=1}^{2^n} v_i\log v_i$ for $\alpha=1$). Therefore it must be strictly smaller. Thus there is a distribution $X\not\in\cY$ such that 
$\mathbf{H}_{\alpha}^{\textup{Metric},\mathrm{det} \{0,1\},0}(X) \geqslant k$ for some small number $\epsilon_n$. 
Since $\PrD{X}\not\in \cY$, it can be strictly separated from $\cY$, i.e. for some $[0,1]$-valued function $D$
we have $\mathbf{E}D(X) - \mathbf{E}D(Y) \geqslant \epsilon_n$ for all $Y\in \cY$. But this function can be simulated
by a randomized boolean circuits with arbitrary small absolute error (coming from a finite precision of the computation),
let us say with error at most $\frac{1}{2}\epsilon_n$. It remains to observe, that according to the definition it means
$\mathbf{H}_{\alpha}^{\textup{Metric},\mathrm{rand},\{0,1\},\epsilon_n/2}(X) < k$.
\end{proof}
\noindent This result does not show how large the gap for the metric entropy, being seen by a deterministic or randomized adversary, can be. It is not even clear if there is a difference between a \emph{deterministic unbounded} circuits and \emph{efficient but randomized} ones. We provide concrete separation results for the two cases: the Shannon and the collision entropy.

\paragraph{Colision Entropy - a gap between all deterministic circuits and efficient randomized ones}

\begin{theorem}
For every $k\leqslant n-2$, there exists a random variable $X\in\{0,1\}^n$ such that 
$\mathbf{H}_{2}^{\textup{Metric},\mathrm{det},\{0,1\},0}(X) \geqslant k$ but $\mathbf{H}_{2}(X) \leqslant k - \bigOm{ \log k }$. Moreover, we have $\mathbf{H}_{2}^{\textup{Metric}, \mathrm{rand}\{0,1\},n+\poly{k},\bigTheta{2^{-k}}}(X) \leqslant k - \bigOm{\log k}$.
\end{theorem}
\begin{remark}
The proof gives us actually the separation even between deterministic and real-valued circuits.
\end{remark}
\begin{proof}
Fix a number $k\leqslant n-2$. For $d=1,\ldots,2^k$ let $D$ be a boolean function such that $|D| = d$ and $p(d) = p_D$ where $p_D$ is given by (\ref{eq:ColisionEntropy_ExtremeDistributions}).
The sequence $p(d)$ is well defined as the solutions $p_D$ of (\ref{eq:RenyiEntropy_ExtremeDistributions}) depend only on $|D|$. Let $\gamma(d) = p(d)\cdot d$. Then
\begin{equation}
 \gamma(d) = 2^{-n}d + \sqrt{\left(2^n - d\right) d \left(2^{-k-n} - 2^{-2n} \right) },
\end{equation}
Consider the set $S = \left\{x\in \{0,1\}^n:\ x = \left(w,0^{n-k}\right) \textup{ for some } w\in\{0,1\}^k \right\}$ (an injection of $\{0,1\}^k$ into $\{0,1\}^n$). Enumerate its elements by $x^{1},x^{2},\ldots $ where $x^{d}$ starts with the $k$-digit binary expansion of $d-1$ and define
\begin{equation}
 \PrD{X}\left(x^{1}\right) = \gamma_1,\quad \PrD{X}\left(x^{d}\right) = \gamma\left(d\right) - \gamma\left(d-1\right) \textup{ for } d=2,\ldots,2^{k},\quad \PrD{X}(x) = 0 \textup{ if } x\not\in S
\end{equation}
Extend $\gamma(d)$ by the same formula to $d \in \left[1,2^k\right]$. We will make use of the following properties of $\gamma$
\begin{claim}
The function $\gamma(d)$, extended to $d\in \left[1,2^k\right]$, is increasing and concave.
\end{claim}
\begin{proof}
We have $\frac{\partial \gamma}{\partial d}=2^{-n} + \frac{2^n-2 d}{2}\cdot \left(d \left(2^n-d\right) \right)^{-1/2}\cdot A$ and 
 $\frac{\partial^2 \gamma}{\partial d^2} = -4^{n-1}\left(d \left(2^n-d\right)\right)^{-3/2}\cdot A$, where $A = \left(2^{-k-n} - 2^{-2n}\right)^{1/2}$. Thus, $\gamma$ is increasing if $d\leqslant 2^{n-1}$ and concave for $d\leqslant 2^{n}$. 
\end{proof}
\noindent Since $\gamma(d)$ decreases with $d$ and $\gamma\left(2^k\right) = 1$, it follows that $\PrD{X}$ is a probability measure on $\{0,1\}^n$. Next, we calculate the metric colision entropy and the colision entropy of $X$.   
\begin{claim}
We have $\mathbf{H}_{2}^{\textup{Metric},\mathrm{det} \{0,1\},0}(X) \geqslant k$. 
\end{claim}
\begin{proof}
Since $\gamma(d)$ is a concave function, the sequence $\gamma(d) - \gamma(d-1) = \PrD{X}\left(x^{d}\right)$ is decreasing. Using this, for any boolean function $D$ we obtain
\begin{align}
 \mathbf{E}D(X) = & \sum_{x}\PrD{X}(x) \cdot D(x) \\
 \leqslant & \max\limits_{I \subset \{0,1\}^n:\, |I| = |D|} \sum\limits_{i \in I}^{d}\PrD{X}\left(x^{i}\right) \\
 \leqslant & \sum\limits_{i=1}^{d} \PrD{X}\left(x^i\right) \\
 = & \gamma_{d} = p(d)\cdot |D|
\end{align}
and by the characterization in Corollary \ref{MetricColisionEntropy_Characterization}, the claim follows. 
\end{proof}
\begin{claim}
We have $\mathbf{H}_{2}(X) \leqslant k - \bigOm{\log k}$ 
\end{claim}
\begin{proof}
Observe that
\begin{equation}\label{eq:ColisionEntropy_FoolingDistribution}
 \PrD{X}\left(x^{d}\right) = 2^{-n} + \left(\sqrt{d\left(2^n-d\right)} - \sqrt{(d-1)\left(2^n-d+1\right)} \right)
\sqrt{2^{-k-n} - 2^{-2n} }.
\end{equation}
Hence,
\begin{align}
 \sum\limits_{x}\PrD{X}(x)^2 = 2\cdot 2^{-n} - 2^{k-2n} + \left(2^{-k-n}-2^{-2n}\right) \sum\limits_{d=1}^{2^k} \left(\sqrt{d\left(2^n-d\right)} - \sqrt{(d-1)\left(2^n-d+1\right)} \right)^2.
\end{align}
Note that 
\begin{align}
 \sqrt{d\left(2^n-d\right)} - \sqrt{(d-1)\left(2^n-d+1\right)} = &
 \frac{2^n-2d+1}{ \sqrt{d\left(2^n-d\right)} + \sqrt{(d-1)\left(2^n-d+1\right) } } \\
 = & \bigTheta{d^{-1/2}\left(2^n-d\right)^{-1/2} \left( 2^n-2d \right)  }
\end{align}
Using this we obtain
\begin{align}
 \sum\limits_{x}\PrD{X}(x)^2 = & 2\cdot 2^{-n} - 2^{k-2n}  + \left(2^{-k-n}-2^{-2n}\right) \bigTheta{\sum\limits_{d=1}^{2^k} \left(2^n-d\right)^{-1}d^{-1}\left(2^n-2d\right)^2} \\
 = & \bigTheta{2^{-n}} +  \bigTheta{2^{-k-n}} \bigTheta{ \sum\limits_{d=1}^{2^k} \left(  d^{-1}\left(2^n-d\right)+\left(2^n-d\right)^{-1}d -2 \right) } \\
 = & \bigTheta{2^{-n}} +  \bigTheta{2^{-k-n}}\bigTheta{2^{n}k} = \bigTheta{2^{-k}k}.
\end{align}
Since $\bigTheta{2^{-k}k} = 2^{-k + \bigTheta{\log k} + \bigTheta{1}}$, the result follows.
\end{proof}
\noindent By combining the last two claims we obtain the first part of the theorem.

\begin{claim}
We have $\mathbf{H}_{2}^{\textup{Metric}, \mathrm{rand}\{0,1\},n+\poly{k},\bigTheta{2^{-k}}}(X) \leqslant k - \bigOm{\log k}$ 
\end{claim}
\begin{proof}
Let $D$ be a real valued (!) function defined as $D(x) = \PrD{X}(x)$. Let $c$ be a positive constant (to be determined later). For every distribution $Y$ over $\{0,1\}^n$ satisfying $\mathbf{H}_2(Y) \geqslant k - c\log k$, by applying the Cauchy-Schwarz Inequality and using the estimate on the colision entropy of $X$, we obtain
\begin{align}
 \mathbf{E} D(X) - \mathbf{E}D(Y) = & \sum\limits_{x}\PrD{X}(x)\cdot D(x) - \sum\limits_{x}\PrD{Y}(x)\cdot D(x) \\
 \geqslant & \sum\limits_{x}\PrD{X}(x)\cdot D(x) - \left(\sum\limits_{x}\PrD{Y}(x)^2 \right)^{1/2} \left(\sum\limits_{x}D(x)^2 \right)^{1/2} \\
 = & \sum\limits_{x} \PrD{X}(x)^2 - \left(\sum\limits_{x}\PrD{Y}(x)^2 \right)^{1/2} \left(\sum\limits_{x}\PrD{X}(x)^2 \right)^{1/2} \\
 = & 2^{-\mathbf{H}_2(X)} \left(1- 2^{\mathbf{H}_2(X)/2 - \mathbf{H}_2(Y)/2}\right) \\
 = & \bigTheta{2^{-k} k}\left(1 - 2^{-\bigTheta{\log k} + c\log k} \right) 
\end{align} 
which is $\bigTheta{2^{-k} k}$ provided that $c$ is sufficiently small. We will show how to simulate $D$ with a randomized efficient boolean circut $D'$. Let $\ell$ be chosen so that $2^{-\ell} \ll 2^{-k}k$, for instance $\ell = \bigOm{k}$. 
\begin{algorithm}
\caption{Distinguishing between $X$ and distributions $Y$ with $\mathbf{H}_2(Y)\geqslant k-\bigTheta{\log k}$}
\begin{algorithmic}[1]
\Require{ $x\in\{0,1\}^n$}
\Ensure{$D'(x)$}
\If{ $x\not\in S$}
\State \Return 0
\Else
\State $d\gets $ a number such that $x = x^{d}$
\For{$j \gets 1,\ldots,\ell$}
\State $r_j:=$ the $j$-th digit of the binary expansion of $\PrD{X}(x^{d})$ 
\State $b_j \gets \{0,1\}$ at random (flip a coin)
\EndFor
\State $j\gets $ the smallest number such that $b_j = 1$ or $0$ if does not exist
\State \Return $r_j$
\EndIf
\end{algorithmic}
\end{algorithm} 

\noindent It is easily seen that for $x=x^{d}$ we have $\mathbf{E}_{b_1,\ldots,b_{\ell} \leftarrow \{0,1\}^{\ell}}{D'(x)} = \sum\limits_{j=1}^{\ell}2^{-j}r_j$.
Therefore, for every $x\in\{0,1\}^n$ we have $\left| \mathbf{E}_{b_1,\ldots,b_{\ell} \leftarrow \{0,1\}^k}{D'(x)} - D(x) \right| \leqslant 2^{-\ell-1} $. Since $b_j$ are indepdent from $X$ and $Y$ it follows that
\begin{equation}
 \mathbf{E}D'(X) - \mathbf{E}D'(Y) \geqslant \mathbf{E}D'(X) - \mathbf{E}D'(Y) - 2^{-\ell-1}.
\end{equation}
Hence, for all $Y$ over $\{0,1\}^n$ with $\mathbf{H}_2(Y)\geqslant k$ we have
\begin{equation}
 \mathbf{E}D'(X) - \mathbf{E}D'(Y) = \bigTheta{2^{-k}k}.
\end{equation}
Finally, note that the complexity of $D'$ is at most $\bigO{n + \poly{\ell}} = \bigO{n + \poly{k}} $.
\end{proof}
\end{proof}
\begin{corollary}
There exists a random variable $X\in\{0,1\}^n$ such that:
\begin{enumerate}[label = (\roman{*})]
 \item $X$ has the collision metric entropy $k = \bigTheta{\log n}$ against all deteterministic boolean functions, with $\epsilon = 0$
 \item $X$ has the collision metric entropy $k - \bigOm{\log \log n}$ against randomized circuits of size $\bigO{n}$, with $\epsilon = \poly{1/n} $
\end{enumerate}
\end{corollary}

\paragraph{Shannon entropy- even larger gap}\ 

\noindent For the Metric Shannon Entropy we provide the following even more stronger separation between randomized and deterministic distinguishers for the Shannon Entropy:
\begin{corollary}
For some absolute constant $c \in (0,1)$, for every $n$ there exists a random variable $X\in\{0,1\}^n$ such that:
\begin{enumerate}[label = (\roman{*})]
 \item Metric Shannon Entropy of $X$ is $k \geqslant c n$, against all deteterministic boolean functions and $\epsilon = 0$
 \item Metric Shannon Entropy of $X$ is $k - \bigOm{n}$, against all randomized circuits and $\epsilon = \bigOm{1}$.
\end{enumerate} 
\end{corollary}
The proof is long and requires a lot of technical calculations, thus is left to the Appendix.

\section{Conclusions}
We developed a new ``geometric" way of looking at metric-type computational entropy and show that it can be usefull in some important situations, especially for the leakage-resilent cryptography. Although the tools of convex analysis seems to be complicated and unintuitive, they can yield some powerfull results as we demonstrated having solved the problem of the derandomization of generalized metric-type entropy. We believe that this nonstandard approach can be helpful in improving our understanding of the computational entropy.

\section{Acknowledgements}
I would like to express special thanks to Stefan Dziembowski and Krzysztof Pietrzak, for their helpful suggestions and discussions.

\bibliographystyle{amsalpha}
\bibliography{./ComputationalEntropy}

\providecommand{\bysame}{\leavevmode\hbox to3em{\hrulefill}\thinspace}
\providecommand{\MR}{\relax\ifhmode\unskip\space\fi MR }
\providecommand{\MRhref}[2]{%
  \href{http://www.ams.org/mathscinet-getitem?mr=#1}{#2}
}
\providecommand{\href}[2]{#2}
\begin{thebibliography}{RTTV08}

\bibitem[BSW03]{Barak2003}
Boaz Barak, Ronen Shaltiel, and Avi Wigderson, \emph{Computational analogues of
  entropy.}, RANDOM-APPROX (Sanjeev Arora, Klaus Jansen, José D.~P. Rolim, and
  Amit Sahai, eds.), Lecture Notes in Computer Science, vol. 2764, Springer,
  2003, pp.~200--215.

\bibitem[CKLR11]{Chung2011}
Kai-Min Chung, Yael~Tauman Kalai, Feng-Hao Liu, and Ran Raz, \emph{Memory
  delegation}, Cryptology ePrint Archive, Report 2011/273, 2011,
  \url{http://eprint.iacr.org/}.

\bibitem[DORS08]{Dodis2008}
Yevgeniy Dodis, Rafail Ostrovsky, Leonid Reyzin, and Adam Smith, \emph{Fuzzy
  extractors: How to generate strong keys from biometrics and other noisy
  data}, SIAM J. Comput. \textbf{38} (2008), no.~1, 97--139.

\bibitem[DP08]{Dziembowski2008}
Stefan Dziembowski and Krzysztof Pietrzak, \emph{Leakage-resilient cryptography
  in the standard model}, IACR Cryptology ePrint Archive \textbf{2008} (2008),
  240.

\bibitem[FR12]{Fuller2011}
Benjamin Fuller and Leonid Reyzin, \emph{Computational entropy and information
  leakage}, Cryptology ePrint Archive, Report 2012/466, 2012,
  \url{http://eprint.iacr.org/}.

\bibitem[GW10]{GentryWichs2010}
Craig Gentry and Daniel Wichs, \emph{Separating succinct non-interactive
  arguments from all falsifiable assumptions}, Cryptology ePrint Archive,
  Report 2010/610, 2010, \url{http://eprint.iacr.org/}.

\bibitem[HILL99]{HILL99}
Johan Hastad, Russell Impagliazzo, Leonid~A. Levin, and Michael Luby, \emph{A
  pseudorandom generator from any one-way function}, SIAM J. Comput.
  \textbf{28} (1999), no.~4, 1364--1396.

\bibitem[KPW13]{Pietrzak2012}
Stephan Krenn, Krzysztof Pietrzak, and Akshay Wadia, \emph{A counterexample to
  the chain rule for conditional hill entropy}, Theory of Cryptography (Amit
  Sahai, ed.), Lecture Notes in Computer Science, vol. 7785, Springer Berlin
  Heidelberg, 2013, pp.~23--39.

\bibitem[Rey11]{Reyzin2011}
Leonid Reyzin, \emph{Some notions of entropy for cryptography}, Information
  Theoretic Security (Serge Fehr, ed.), Lecture Notes in Computer Science, vol.
  6673, Springer Berlin Heidelberg, 2011, pp.~138--142.

\bibitem[RTTV08]{Reingold2008}
Omer Reingold, Luca Trevisan, Madhur Tulsiani, and Salil Vadhan, \emph{Dense
  subsets of pseudorandom sets}, Proceedings of the 2008 49th Annual IEEE
  Symposium on Foundations of Computer Science (Washington, DC, USA), FOCS '08,
  IEEE Computer Society, 2008, pp.~76--85.

\bibitem[RW]{Renner2004}
R.~Renner and S.~Wolf, \emph{{Smooth Renyi entropy and applications}},
  International Symposium on Information Theory, 2004. ISIT 2004. Proceedings.,
  IEEE, p.~232.

\bibitem[Ser74]{Serfling1974}
R.~J. Serfling, \emph{{Probability Inequalities for the Sum in Sampling without
  Replacement}}, The Annals of Statistics \textbf{2} (1974), no.~1, 39--48.

\bibitem[Sha48]{Shannon1948}
C.~E. Shannon, \emph{{A mathematical theory of communication}}, Bell system
  technical journal \textbf{27} (1948).

\bibitem[VZ12]{Vadhan2012}
Salil Vadhan and Colin~Jia Zheng, \emph{Characterizing pseudoentropy and
  simplifying pseudorandom generator constructions}, Proceedings of the 44th
  symposium on Theory of Computing (New York, NY, USA), STOC '12, ACM, 2012,
  pp.~817--836.

\bibitem[Yao82]{Yao1982}
Andrew~C. Yao, \emph{Theory and application of trapdoor functions}, Proceedings
  of the 23rd Annual Symposium on Foundations of Computer Science (Washington,
  DC, USA), SFCS '82, IEEE Computer Society, 1982, pp.~80--91.

\end{thebibliography}

\appendix 

\section{Proofs}

\begin{reptheorem}{MetricHILL_Conversion}
Let $\cP$ be the set of all probability measures over $\Omega$. Suppose that we are given a class $\mathcal{D}$ of $[0,1]$-valued functions on $\Omega$, with the following property: if $D\in \cD$ then $D^{c}=^{\textup{def}}\mathbf{1}-D \in \cD$. For $\delta > 0$, let $\mathcal{D'}$ be the class consisting of all convex combinations of length $\mathcal{O}\left(\frac{\log|\Omega|}{\delta^2}\right)$ over $\mathcal{D}$. Let $\mathcal{C}\subset \cP$ be any arbitrary convex subset of probability measures and $X \in \cP$ be a fixed distribution. Consider the following statements:
\begin{enumerate}[label = \roman{*}]
\item \label{itm:HILLMetricCoversion_a} $X$ is $\left(\mathcal{D},\epsilon+\delta\right)$ indistinguishable from \emph{some} distribution $Y\in \mathcal{C}$ (HILL Entropy)
 \item \label{itm:HILLMetricCoversion_b} $X$ is $\left(\mathcal{D'},\epsilon\right)$ indistinguishable from the set of \emph{all} distribution $Y\in \mathcal{C}$ (Metric Entropy)
\end{enumerate}
Then (\ref{itm:HILLMetricCoversion_b}) implies (\ref{itm:HILLMetricCoversion_a}).  
\end{reptheorem}
\begin{proof}
This result  was formulated in \cite{Barak2003} in a less general form, namely $\Omega = \{0,1\}^{n}$, $\mathcal{C}$ is the set of distributions with min-entropy at least $k$, and $\mathcal{D},\mathcal{D'}$ are the classes of $[0,1]$-valued circuits of size $s$ and $\mathcal{O}\left( s \cdot \frac{n}{\delta^2} \right)$ respectively. The inspection of the proof shows that: 
(a) the chosen space $\Omega$ can be an arbitrary finite set, and the number $n$ appearing in the assertion is equal to $\log|\Omega|$, (b) the chosen set $\mathcal{C}$ can be replaced by an arbitrary convex set of distributions, (c) the complexity of the class $\mathcal{D'}$ is chosen only to ensure that $\mathcal{D'}$ contains all convex combinations of length $\mathcal{O}\left(\frac{\log |\Omega|}{\delta^2} \right)$ of elements of $\mathcal{C}$.
\end{proof}

\subsection{Separation of Metric and Smooth Entropy}

\begin{replemma}{hard_functions_1}\label{app:hard_functions_1}
For any $C \geqslant 1$ and sufficiently large $\ell$ there exists a boolean function $f$ over $\{0,1\}^{\ell}$, such that $\mathrm{bias}(f) = 1-2^{1-C}$ and for all circuits $D$ of size $\bigO{\nicefrac{2^{\ell-C}\delta^2}{\ell-C-2\log(1/\delta}}$ we have
\begin{equation*}
 \Pr{x\leftarrow f^{-1}(\{1\})}{D(x) = f(x)} + \Pr{x\leftarrow f^{-1}(\{0\})}{D(x) = f(x)} < 1+\delta
\end{equation*}
\end{replemma}
\begin{proof}
Chose a set $A$ by sampling $m=2^{\ell-C}$ elements $x\in\{0,1\}^{\ell}$ without replacement. The random variables $D(x)$ for $x\in A$ are not independent. However, the Hoeffding Inequality still holds for sampling without replacement and gives us
\begin{equation}
 \Pr{A}{\mathbf{E}D\left(U_A\right) - \mathbf{E}D(U) > \frac{1}{2}\delta} \leqslant \exp\left( -\delta^2 2^{\ell-C} \right).
\end{equation}
Let $B = A^{c}$. Since the set $B$ can be viewed as chosen by sampling $2^{\ell}-2^{\ell-C}$ elements from $\{0,1\}^{\ell}$ without replacement, applying the Hoeffding Inequality again, we have
\begin{equation}
 \Pr{B}{\mathbf{E} D^{c}\left(U_B\right) - \mathbf{E}D^{c}(U) > \frac{1}{2}\delta} \leqslant \exp\left( -\delta^2 2^{\ell}\left(1-2^{-C}\right) \right)
\end{equation}
Therefore, for every fixed circuit $D$ the inequality
\begin{equation*}
  \mathbf{E}D\left(U_A\right) + \mathbf{E}D^{c}\left(U_B\right) > 1+\delta
\end{equation*}
holds with probability at most $2\exp\left(-2^{\ell-C}\delta^2\right)$ over choosing $A,B$.
By a union bound over all $\exp(\bigO{s\log s}) < \frac{1}{2}\exp\left(2^{\ell-C}\delta^2\right)$ circuits of size $s$, we obtain that there exists set $A$ and $B = A^{c}$ such that for \emph{every} circuit $D$ of size $s$ we have
\begin{equation*}
 \mathbf{E}D\left(U_A\right) + \mathbf{E}D^{c}\left(U_B\right) \leqslant 1+\delta
\end{equation*}
We define $f$ to be $\I{A}$ and the proof is finished.
\end{proof}

\begin{remark}
If the assertion of the lemma\ref{hard_functions_1} is satisfied by a function $f$ then also by $1-f$. Since $\mathrm{bias}(f) = 1-2\cdot 2^{-C}$, replacing $f$ with $1-f$ if necessary, we may assume that $\#\left\{x:\, f(x)=1\right\} = \left(\frac{1}{2}-\frac{1}{2}\mathrm{bias}(f) \right)2^{\ell} = 2^{\ell-C}$. This in turn implies that for all circuits $D$
\begin{align*}
 \Pr{x}{f(x) = D(x)}  & = \\
 = &\ 2^{-C}\Pr{x\leftarrow f^{-1}(\{1\})}{D(x) = 1} + \left(1-2^{-C}\right) \Pr{x\leftarrow f^{-1}(\{0\})}{D(x) = 0} \\
 < &\ 1-2^{-C} + \delta = \frac{1}{2}+\frac{1}{2}\mathrm{bias}(f) + \delta.
\end{align*}
Thus, we have retrieved the classical result on $\delta$-hard functions, as for every function $f$, the value of $f(x)$ can be guessed trivially (using a constant function) for at least $\frac{1}{2}+\frac{1}{2}\mathrm{bias}(f)$ fraction of inputs $x$.
\end{remark}

\begin{reptheorem}{Separation_MetricSmooth_Conditional}\label{app:Separation_MetricSmooth_Conditional}
For sufficiently large $n$, and for any $C > 0$, $k<n-C$ and $\epsilon > 0$ there exists a pair of jointly distributed random variables $X\in \{0,1\}^n$, $Z\in\{0,1\}^m$ such that 
\begin{enumerate}[label = (\roman{*})]
 \item $\mathbf{H}^{1/2}_{\infty}(X|Z) \leqslant k+1$
 \item $\mathbf{H}^{\text{Metric},\det  [0,1], s, \epsilon}(X|Z) \geqslant k+C$ for $ s =  \bigOm{\frac{2^{k+m}\epsilon^4}{(k+m)\log \left(2^{k+m}\epsilon^2 \right)}}$
\end{enumerate}
\end{reptheorem}
\begin{proof}
Fix a distribution $Z$ over $\{0,1\}^m$. For every $z$, chose a $2^k$-element subsets $A(z)$ of $S=\{0,1\}^{k+C}$. Let $B(z)= A(z)^c$. Let $X$ be a distribution (jointly distributed with $Z$) such that $X|Z=z$ is uniform over $A(z)$. We observe that $\mathbf{H}_{\infty}(X|Z)=k$ and $\mathbf{H}_{\infty}^{1/2}(X|Z) \leqslant k+1$, since $\Delta\left(X|Z=z,Y_z\right) \leqslant 1/2$ for every distribution $Y_z$ over $\{0,1]\}^n$ such that $\mathbf{H}_{\infty}\left(Y_z\right)\geqslant k+1$. 
Let $Y|Z=z$ be uniform over $S$. Assuming that $\Hmtr{X|Z}{\{0,1\},s,\epsilon} < k+C$ with $\epsilon = \left(1-2^{-C}\right)\delta$, from the definition of Metric Entropy (replacing $D$ with $D^c$ if necessary) we obtain for some $D$ of size $s$
\begin{align}
 \epsilon \leqslant & \ \mathbf{E}D(X,Z) - \mathbf{E}D\left(U_S\times Z\right) = \nonumber \\
  = & \ \mathbf{E}_{z\leftarrow Z}\left[ \mathbf{E}D(X|Z=z,z) - \mathbf{E}D\left(U_S, z\right)  \right] \nonumber \\
  = & \ \mathbf{E}_{z\leftarrow Z}\left[ \mathbf{E}D\left(U_{A(z)},z\right) - \frac{\left|A(z)\right|}{|S|} \mathbf{E}D\left(U_{A(z)}\right) - \frac{|B(z)|}{|S|}\mathbf{E}D\left(U_{B(z),z},z\right)  \right] \nonumber \\
  = & \mathbf{E}_{z\leftarrow Z}\left[\left(1-2^{-C}\right) \mathbf{E}D\left(U_{A(z)},z\right) - \left(1-2^{-C}\right)\mathbf{E}D\left(U_{B(z),z},z\right) \right] \nonumber \\
  = & \left(1-2^{-C}\right)\mathbf{E}_{z\leftarrow Z}\left[ \mathbf{E}D\left(U_{A(z)},z\right) + \mathbf{E}D^{c}\left(U_{A(z)},z\right) - 1 \right]
\end{align}
Therefore, for every distribution $Z$ there exists
a circuit $D$ of size $s$ such that
\begin{equation}
 \mathbf{E}_{z\leftarrow Z}\left[ \ExV{x\leftarrow A(z)}{D(x,z)} + \ExV{x\leftarrow B(z)}{D^{c}(x,z)} \right] \geqslant 1+\delta
\end{equation}
by a min-max theorem and obtain that there exists a circuit $D$ (not efficient itself but being a convex combination of circuits of size $s$) such that
\begin{equation}
 \textup{for all distributions }Z : \quad \mathbf{E}_{z\leftarrow Z}\left[ \ExV{x\leftarrow A(z)}{D(x,z)} + \ExV{x\leftarrow B(z)}{D^{c}(x,z)} \right] \geqslant 1+\delta.
\end{equation}
By a standard approximation via Chernoff Bounds, for some circuit of size $s'=\bigO{\nicefrac{(k+m) s}{\epsilon^2}}$ we get
\begin{equation}
 \textup{for all distributions }Z : \quad \mathbf{E}_{z\leftarrow Z}\left[ \ExV{x\leftarrow A(z)}{D(x,z)} + \ExV{x\leftarrow B(z)}{D^{c}(x,z)} \right] \geqslant 1+\delta/2.
\end{equation}
Especially, for every $z$ we obtain
\begin{equation}\label{ConditionalSeparation_CircuitGuess}
  \ExV{x\leftarrow A(z)}{D(x,z)} + \ExV{x\leftarrow B(z)}{D^{c}(x,z)}  \geqslant 1+\delta/2
\end{equation}
Observe that this inequality is valid independently on the choice of $A(z)$. We argue, that if $A(z)$ are chosen at random, this inequality becomes a `hard task' for small circuits. More precisely, we make use of the following lemma on hard functions
\begin{lemma}\label{ConditionalHardFunction}
Let $\delta \in (0,1)$, $C>0$ and $\cD$ be a class of boolean randomized functions on $\{0,1\}^{\ell+m}$ of cardinality at most $\exp\left(c\cdot 2^{\ell+m-C}\delta^2\right)$ for universal constant $c$. Then there exists a function $f$ on $\{0,1\}^{\ell+m}$ such that $\mathrm{bias}(f(\cdot,z)) = 1-2^{1-C}$ for every $z$, with the following property: for every $D\in\cD$ there exists at least one $z$ satisfying
\begin{equation}
 \Pr{x:\ f(x,z) = 1}{D(x,z) = f(x,z)} + \Pr{x:\ f(x,z) = 0}{D(x,z) = f(x,z)} < 1+\delta.
\end{equation}
\end{lemma}
\begin{proof}
Fix a function $D\in\cD$. For every $z$ chose a set $A(z)$ by sampling $m=2^{\ell-C}$ elements $x\in\{0,1\}^{\ell}$ without replacement. The random variables $D(x,z)$ for $x\in A$ are not independent. However, the Hoeffding Inequality holds for sampling without replacement (see \cite{Serfling1974} for instance) and gives us
\begin{equation}\label{ineq:ConditionalHardFunctions_Chernoff1}
 \Pr{A(z)}{\mathbf{E}D\left(U_{A(z)},z\right) - \mathbf{E}D(U,z) \geqslant \delta/2} \leqslant \exp\left( -\bigOm{\delta^2 2^{\ell-C}} \right).
\end{equation}
Let $B(z) = A(z)^{c}$. Since the set $B(z)$ can be viewed as chosen by sampling $2^{\ell}-2^{\ell-C}$ elements from $\{0,1\}^{\ell}$ without replacement, applying the Hoeffding Inequality again, we have
\begin{equation}\label{ineq:ConditionalHardFunctions_Chernoff2}
 \Pr{B(z)}{\mathbf{E} D^{c}\left(U_{B(z)},z\right) - \mathbf{E}D^{c}(U,z) \geqslant \delta/2} \leqslant \exp\left( -\bigOm{\delta^2 2^{\ell}\left(1-2^{-C}\right)} \right)
\end{equation} 
Since $\mathbf{E}D(U) + \mathbf{E}D^{c}(U) = 1$, inequalities (\ref{ineq:ConditionalHardFunctions_Chernoff1}) and (\ref{ineq:ConditionalHardFunctions_Chernoff2}) for every $z$ yield 
\begin{equation}
 \Pr{A(z),B(z)}{\mathbf{E}D\left(U_{A(z)},z\right) + \mathbf{E} D^{c}\left(U_{B(z)},z\right) \leqslant 1+\delta} \geqslant
 1-\exp\left( -\bigOm{\delta^2 2^{\ell-C}} \right)
\end{equation}
Thus, probability that all values $z$ are `bad' is equal to
\begin{equation}
 \Pr{A,B}{\text{for every } z:\  \mathbf{E}D\left(U_{A(z)},z\right) + \mathbf{E} D^{c}\left(U_{B(z)},z\right) \leqslant 1+\delta} \leqslant
 \exp\left( -\bigOm{\delta^2 2^{\ell+m-C}} \right),
\end{equation}
and by a union bound over all members of $\cD$ the result follows.
\end{proof}
\noindent Note that condition $\mathrm{bias}{f} = 1-2^{1-C}$ in Lemma \ref{ConditionalHardFunction} means that either $|f| = 2^{C}$ or $\left|f^{c}\right| = 2^{C}$. Clearly, the lemma is valid also for $f^c$. Thus, without losing generality, let $|f| = 2^{C}$. Apply Lemma \ref{ConditionalHardFunction} to $\ell = k + C$.
Define the sets $A(z)$ as $A(z) = \left\{x:\,f(x,z) = 1\right\}$ and $B(z) = \left\{x:\, f^{c}(x,z) = 1\right\}$.
Since $s' = \bigO{\nicefrac{(k+m)s}{\epsilon^2}}$, inequality \ref{ConditionalSeparation_CircuitGuess} contradicts the lemma provided that 
\begin{equation}
  \exp\left(s'\log s'\right)\exp\left(-\bigOm{2^{k+m}\epsilon^2 }\right) < 1
\end{equation}
or in other words if
\begin{equation}
 \frac{(k+m)s}{\epsilon^2} < c\cdot \frac{2^{k+m}\epsilon^2}{\log \left(2^{k+m}\epsilon^2 \right)},
\end{equation}
which is equivalent to 
\begin{equation}
s < c \frac{2^{k+m}\epsilon^4}{(k+m)\log \left(2^{k+m}\epsilon^2 \right)}
\end{equation}
\end{proof}

\subsection{Characterizations of R\'{e}nyi Metric Entropy}

\begin{replemma}{RenyiEntropy_ExtremeDistributions}\label{app:RenyiEntropy_ExtremeDistributions}
Let $\alpha>1$ be fixed, let $D:\{0,1\}^{n}\rightarrow \{0,1\}$ be a function and $\mathcal{Y}_k = \left\{Y\in\{0,1\}^{n}:\, \mathbf{H}_{\alpha}(Y)\geqslant k \right\}$. Then 
\begin{equation}
 \max\limits_{Y\in \mathcal{Y}_k} \mathbf{E}D(Y) = 
\left\{
\begin{array}{rl}
 p_D \cdot |D|,  & \textup{if } |D| < 2^k \\
 1, & \textup{otherwise}
\end{array}
\right.
\end{equation}
where $p_D$, for $|D| \leqslant 2^k$, is the greatest number satisfying the following system 
\begin{equation}\label{app:eq:RenyiEntropy_ExtremeDistributions}
\left\{
\begin{array}{rcl}
 p_D^{\alpha}|D| + q_D^{\alpha}|D^c| &=& 2^{-(\alpha-1)k}  \\
 p_D |D| + q_D |D^c| &=&  1 \\
 p_D,q_D & \geqslant & 0
\end{array}\right. 
\end{equation}
Moreover, the solution $p_D$ is unique provided that $k < n-1$.
\end{replemma}
\begin{proof}
First we prove that $\max\limits_{Y\in\mathcal{Y}_k} \mathbf{E}D(Y) < 1$ is equivalent to $|D|<2^k$. 
Suppose that $\max\limits_{Y\in\mathcal{Y}_k} \mathbf{E}D(Y) < 1$. If $\left| D^{-1}(1) \right| \geqslant 2^k$ then for $Y$ being uniform over $D^{-1}(1)$ we get a contradiction as $\mathbf{E}D(Y) = 1$ and $\mathbf{H}_{\alpha}(Y)\geqslant \Hmin{Y} \geqslant k$. The other direction follows from the following Lemma, proved in the Appendix:
\begin{lemma}\label{HighEntropy_LargeSupport}
Let $X\in\{0,1\}^n$ be a random variable satisfying $\mathbf{H}_{\alpha}(X)\geqslant k$.
Then $\left| \mathrm{supp}\left(\PrD{X}\right) \right| \geqslant 2^k$. 
\end{lemma}
\noindent Assume that $\max\limits_{Y\in\mathcal{Y}_k} \mathbf{E}D(Y) < 1$ and let $Y$ be a distribution maximizing $\mathbf{E}D(\cdot)$ over the set $\mathcal{Y}_k$. We will show, that $Y$ may be assumed to be uniform if conditioned on the sets $D^{-1}(0)$ and $D^{-1}(1)$. The first part is clear because modifying the distribution $Y$ outside the support of $D$ we do not change the value $\mathbf{E}D(X)$. To prove the second one,
define $\PrD{Y'}(x)$ to be $\frac{1}{\left|D^{-1}(1)\right|}\sum\limits_{x'\in D^{-1}(1)} \PrD{Y}(x') $ if $x\in D^{-1}(1)$ and $\PrD{Y}$ otherwise. By Jensen's inequality we get $\sum\limits_{x\in D^{-1}(1)} \PrD{Y'}(x)^{\alpha} \leqslant \sum\limits_{x\in D^{-1}(1)} \PrD{Y}(x)^{\alpha}$ and thus $\mathbf{H}_{\alpha}(Y') \geqslant \mathbf{H}_{\alpha}(Y)$. Since $D$ is boolean, we also have $\mathbf{E}D(Y') = \mathbf{E}D(Y)$. Therefore,
for some $p=p_D,q=q_D$ we have
\begin{equation}
 \PrD{Y}(x) = p\I{D^{-1}(1)}(x) + q\I{D^{-1}(0)}(x),
\end{equation}
where $p,q$ should be chosen so that $Y$ is a proability measure and satisfies the constraint $\mathbf{H}_{\alpha}(Y)\geqslant k$. These two conditions are exactly equations (\ref{eq:RenyiEntropy_ExtremeDistributions}). Note that since the maximizier $Y \in \cY_k$ for $D$ exists, this system certainly has a solution. To prove that this solution is unique, we observe that after substituting $\gamma = p_D |D|$ the first equation becomes $f(\gamma) = 0$ where $f(\gamma) = \gamma^{\alpha} |D|^{1-\alpha} + (1-\gamma)^{\alpha} \left|D^c\right|^{1-\alpha} - 2^{-k(\alpha-1)}$ and $0 \leqslant \gamma \leqslant 1/|D|$. Observe that the function $f$ is strictly convex and, provided that $k < n-1$, we have $f(0) = \left(2^{n}-|D|\right)^{1-\alpha}-2^{-(\alpha-1)k} < 0$. Therefore, there can be at most one solution $\gamma \geqslant 0$. 
\end{proof}

\subsection{Metric Min-Entropy Against Different Distinguishers}

\begin{reptheorem}{MetricWorstCase_passing_to_realvalued}\label{app:MetricWorstCase_passing_to_realvalued}
Let $X\in\{0,1\}^n$ and $Z\in\{0,1\}^n$ be random variables. Then we have $\Hmtr{X|Z}{\textup{det},\{0,1\},s,\epsilon} = \Hmtr{X|Z}{\textup{det},[0,1],s',\epsilon}$ where $s'\approx s$. 
\end{reptheorem}
\begin{proof}
We need only to show that if $\Hmtr{X|Z}{\textup{det},\{0,1\},s,\epsilon}\geqslant k$ then also
$\Hmtr{X|Z}{\textup{det},[0,1],s',\epsilon} \geqslant k$ for $s'\approx s$. Let $\cY$ be the set of distributions of the random variables of the form $(Y,Z)$ where $Y\in\{0,1\}^n$ and 
$\Hmin{Y|Z}\geqslant k$. Suppose, that $\Hmtr{X|Z}{\textup{det},[0,1],s',\epsilon} < k$. According to the definition, there exists a $[0,1]$-valued function $D$ of complexity $s'$ such that
\begin{equation}\label{ineq:MinEntropy_Distinguishing}
 \mathbf{E}D(X,Z) - \max\limits_{\PrD{Y,Z} \in \cY}\mathbf{E}D(Y,Z) \geqslant \epsilon
\end{equation}
We shall show that $D$ can be replaced by a boolean distinguisher $D'$ of (almost) the same complexity.
Let $\PrD{Y_0,Z}$ be a distribution that maximizies $\mathbf{E}D(\cdot)$ over $\cY$. It means that $p_0 = \PrD{Y_0,Z} $ is a solution of the following \emph{constrained optimization problem} in $\mathbb{R}^{2^{n+m}}$:
\begin{equation}\label{eq:MinEntropy_Optimization}
\begin{array}{rl}
 \underset{p}{\text{maximize}} & \sum\limits_{x,z} p(x,z) \cdot D(x,z)  \\
 \text{s.t.} & \left\{ \begin{array}{rll}
   \sum\limits_{x,z} p(x,z) &  = 1 &   \\
  \sum\limits_{x} p(x,z) & = \PrD{Z}(z), & \text{ for every } z\\
  -p(x,z) & \leqslant     0, &  \text{ for all } x,z \\
  p(x,z)  & \leqslant   2^{-k}\PrD{Z}(z), &  \text{ for all } x,z  
\end{array}  \right.
\end{array}
\end{equation}
where the constraints in this optimization problem describe the set $\cY$. We can assume that $p_0$ is chosen to be flat (for every $x,z$ either $p_0(x,z)/\PrD{Z}(z) = \PrD{Y_0|Z}(x) = 2^{-k}$ or $p_0(x,z) = 0$) as otherwise we would have $p_0 = t p_1 + (1-t)p_2$ where $p_1,p_2\in \cY$ and then either $p = p_1$ or $p = p_2$ gives $\langle D, p \rangle \geqslant \langle D, p_0 \rangle$. The proof will be complete, if we will find a function $D'$ satisfying the following conditions: 
\begin{enumerate}[label=(\emph{\alph{*}}), ref=\alph{*}]
 \item $D'$ is boolean \label{MinEntropy_ChangeDistinguisher_a}
 \item $\PrD{Y_0,Z}$ is a maximizier for $D'$ over $\cY$, (i.e. $\mathbf{E}D'\left(Y_0,Z\right) \geqslant \mathbf{E}D'(Y,Z)$ for all $\PrD{Y,Z} \in \cY$) \label{MinEntropy_ChangeDistinguisher_b}
 \item $\mathbf{E}D'(X,Z) - \mathbf{E}D'(Y_0,Z) \geqslant \epsilon$ \label{MinEntropy_ChangeDistinguisher_c}
 \item $D'$ has the complexity $s$ \label{MinEntropy_ChangeDistinguisher_d}
\end{enumerate}
Consider now the condition in (\ref{MinEntropy_ChangeDistinguisher_b}). It can be rewritten as $\langle D', \PrD{Y_0,Z}-\PrD{Y,Z} \rangle \geqslant 0$ for all $\PrD{Y,Z}\in \cY$ (we indetify functions $D',\PrD{X,Z},\PrD{Y,Z}$ on $\{0,1\}^{2^{n+m}}$ with vectors of $\mathbb{R}^{2^{n+m}}$). 
The set of all such $D'\in \mathbb{R}^{2^{n+m}}$ is the \emph{normal cone of $\cY$ at $\PrD{Y_0,Z}$}. 
\begin{claim}\label{MinEntropy_ConeCondition}
The normal cone of $\cY$ at $p_0=\PrD{Y_0,Z}$, i.e. all real valued functions $D'$ for (\ref{MinEntropy_ChangeDistinguisher_b}), is decribed by the following condition:
there exist the \emph{Lagrange Multipliers}: $\lambda^2(z),\lambda^{3}(x,z),\lambda^{4}(x,z)\geqslant 0$ such that
\begin{equation}\label{eq:MinEntropy_LagrangeMultipliers}
 D'(x,z) = \lambda_2(z) - \lambda_3(x,z) + \lambda_4(x,z)
\end{equation}
and satisfying the so called \emph{complementary slackness condition}: $\lambda_3(x,z),\lambda_4(x,z)$ can be nonzero only if the corresponding costraint is \emph{active}, i.e. if $p_0(x,z) = 0$ or $p_0(x,z) = 2^{-k}\PrD{Z}(z)$ respectively.  
\end{claim}
\begin{proof}
We can replace the first two (equaility-type) constraints by the inequalities $\sum\limits_{x,z} p(x,z) \leqslant 1$ and $\sum\limits_{x} p(x,z) \leqslant \PrD{Z}(z)$, as at the maximizier the equaility will be achieved beacuse of $D'(x,z)\geqslant 0$. Moreover, the first inequality can be dropped as it is implied by the second one. Now, the claim follows by standard facts from convex optimization: the normal cone of a set described by linear inequalities (a polyhedron) is a cone generated by the gradients of the `active' constraints. 
\end{proof}
\noindent It is easy to see, that the above can be stated equivalently as follows:
\begin{claim}\label{MinEntropy_ConeCondition2}
The normal cone of $\cY$ at $\PrD{Y_0,Z}$, consists of all real valued functions $D'$ satisfying 
\begin{equation}
 D'\left(x_1,z\right) \geqslant D'\left(x_2,z\right) \quad \textup{for every  } z,x_1,x_2 \textup{ such that } \PrD{Y_0,Z=z}\left(x_1,z\right) = 2^{-k},\ \PrD{Y_0,Z=z}\left(x_2,z\right) = 0
\end{equation}
\end{claim}
\noindent The definition of $\PrD{Y_0,Z}$ implies that $D$ belongs to the normal cone of $\cY$ at $\PrD{Y_0,Z}$ and thus satisfies the assertion of Claim \ref{MinEntropy_ConeCondition}. From Claim \ref{MinEntropy_ConeCondition2} it follows that
also every treshold of $D$: any function of the form $D'(x,z) = \I{ \left\{ D(x,z) > t \right\}}$ is also in the normal cone. Thus, every such $D'$ satisfies (\ref{MinEntropy_ChangeDistinguisher_a}),(\ref{MinEntropy_ChangeDistinguisher_b}), and (\ref{MinEntropy_ChangeDistinguisher_d}). Finally, since have 
\begin{equation}
 \epsilon \leqslant \mathbf{E}D(X,Z) - \mathbf{E}D\left(Y_0,Z\right) = \int\limits_{t\in[0,1]}\left( \PrD{X,Z}{\left[D(X,Z) > t\right] } - \PrD{X,Z}{\left[D(X,Z) > t\right] }\right),
\end{equation}
for some $t$ the corresponding function $D'$ satisfies also (\ref{MinEntropy_ChangeDistinguisher_b}). This proves the first part of the theorem.

To prove the second one, suppose that there exists a $[0,1]$-valued function $D$ (possibly computationally ineffecient) satisfying (\ref{ineq:MinEntropy_Distinguishing}). As in the proof od the first part, let $\PrD{Y_0,Z}$ be a flat distribution maximizing $\mathbf{E}D(\cdot)$ over $\cY$. Due to the first part of the theorem, we may assume that $D$ is boolean. Suppose now, that $D\left(x_0,z_0\right) = 0$ for some $\left(x_0,z_0\right) \not\in \supp{X,Z}$. By (\ref{MinEntropy_ConeCondition2}), we obtain that $D(x,z) = 0$ for all $(x,z) \in \supp{X,Z}$. But then we have $\mathbf{E}D\left(Y_0,Z\right) = 1$ which contraddicts to (\ref{ineq:MinEntropy_Distinguishing}) as $\epsilon > 0$.
\end{proof}

\begin{reptheorem}{MetricRelaxed_passing_to_realvalued}\label{app:MetricRelaxed_passing_to_realvalued}
Let $X\in\{0,1\}^n$ and $Z\in\{0,1\}^n$ be random variables. Then we have $\mathbf{H}^{\textup{Metric-rlx}, \det[0,1], s,\epsilon}{X|Z}= \mathbf{H}^{\textup{Metric-rlx}, \det[0,1], s',\epsilon}{X|Z}$ where $s'\approx s$. 
\end{reptheorem}
\begin{proof}
The proof follows easily by inspecting the previous proof for the case of the metric min-entropy. Namelly, for the relaxed definition we only need to remove the condition $ \sum\limits_{x} p(x,z)  = \PrD{Z}(z)$ from the description of the optimization problem given by equation \ref{eq:MinEntropy_Optimization}.
\end{proof}

\subsection{Shannon Entropy against different distinnguishers}

\begin{proof}
Enumerate elements of $\{0,1\}^k$ by $x_1,x_2,\ldots $ where $x_d$ is the $k$-digit
binary expansion of $d-1$. We will construct the distribution $X$ explicity in the following way:
for every $d = 1,\ldots, 2^k$ let $p = p(d)$ be a solution of (\ref{eq:ShannonEntropy_ExtremeDistributions}) (we will prove later that this solution is unique). Define the sequence $\gamma(d) = p(d)\cdot d$.
Let $X$ be a distribution on $\{0,1\}^k$ defined by $\PrD{X}(x_1) = \gamma_1$ and $\PrD{X}\left(x_d\right) = \gamma\left(d\right) - \gamma\left(d-1\right)$. To prove that this construction works we need to show that $X$ is a probability measure and satisfy claimed estimates on its entropy and pseudoentropy. This task involves a lot of calculus to study the solutions of (\ref{eq:ShannonEntropy_ExtremeDistributions}). The proof will be divided into subsequently claims.

\begin{claim}
Let $n > 1$ and $k< n-1$. Then for every real number $d\in \left[1,2^k\right]$, the system (\ref{eq:ShannonEntropy_ExtremeDistributions}) has a unique solution $(p,q)=(p(d),q(d))$. Moreover, for $d<2^{k}$, this solution is a smooth function of $d$. 
\end{claim}
\begin{proof}
The proof will be splited into three parts \newline
\underline{The existence of a solution}. First, we parametrize the solutions of the second equation of (\ref{eq:ShannonEntropy_ExtremeDistributions}) by $p(\gamma) = \frac{\gamma}{d}$ and $q(\gamma)= \frac{1-\gamma}{2^n-d}$ for $\gamma\in [0,1]$. Now, the left side of the first equation of (\ref{eq:ShannonEntropy_ExtremeDistributions}), can be viewed as a function $F$ of $\gamma$. Namely, for fixed $d$, define
\begin{equation}\label{eq:ShanonEntropy_ExtremeDistributions_AuxiliaryFunction}
 F(\gamma) = \gamma\log \gamma + (1-\gamma)\log(1-\gamma) - \gamma\log d - (1-\gamma)\log\left(2^{n}-d\right)
\end{equation}
Then the system \ref{eq:ShannonEntropy_ExtremeDistributions} is equivalent to the equation
\begin{equation}\label{eq:ShanonEntropy_ExtremeDistributions_SimplifiedEquation}
 F(\gamma) = -k,\quad \gamma \in [0,1]
\end{equation}
Observe that
\begin{equation}
 F(0) = -\log\left(2^{n}-d\right) <  -k \leqslant -\log d = F(1)
\end{equation}
and therefore, by the Darboux Principle, we conclude that with some $\gamma\in[0,1]$ we have $F(\gamma) = -k$. It follows that there exists numbers $p,q$ being a solution of (\ref{eq:ShannonEntropy_ExtremeDistributions}). 
\newline
\underline{The uniquness and smoothness}. We calculate the derivative of $F$ with respect to $\gamma$:
\begin{align}
 \frac{\partial F}{\partial \gamma} = & \log \gamma - \log(1-\gamma) + 
\log\left(2^n-d\right) - \log d  \\
 = & \log\left(\frac{\gamma}{1-\gamma}\right) - \log\left(\frac{d}{2^n -d}\right).
\end{align}
Hence, the function $F(\gamma)$ increases if $\gamma > \frac{d}{2^{n}}$ and decreases for $\gamma < \frac{d}{2^{n}}$. Since $F(0) < -k$ there cannot be a solution of $F(\gamma) = -k$ for $\gamma < \frac{d}{2^n}$.
Therefore, the solution $p(\gamma),q(\gamma)$ exists only for some $\gamma > \frac{d}{2^n}$ which satisfy $F(\gamma)=-k$ and it is unique as the function $F(\gamma)$ is then increasing. Finally, this number $\gamma = \gamma(d)$ a $C^{\infty}$ function of $d$ by the Inverse Function Theorem, if only $\gamma(d) < 1$ or equivalently if $d<2^{k}$.
\end{proof}
\begin{claim}
Let $p(d),q(d)$ be the unique solution of the system (\ref{eq:ShannonEntropy_ExtremeDistributions}). Define $\gamma(d) = p(d)\cdot d$. Then
\begin{equation}
 \frac{\partial \gamma }{\partial d} = \frac{p-q}{\log p - \log q}
\end{equation}
and
\begin{equation}
 \frac{\partial^2 \gamma}{\partial d^2} = -\frac{\frac{d}{p}\left(\frac{\partial p}{\partial d}\right)^2 + \frac{2^n-d}{q}\left(\frac{\partial q}{\partial d}\right)^2}{\log p(d) - \log q(d)}
\end{equation}
Especially, $\gamma(d)$ is a concave function.
\end{claim}
\begin{proof}
For every $d$ we have $F(\gamma(d)) = -k$. Deriverating this equation with respect to $d$, we obtain
\begin{align}
 0 = \frac{\partial{F(\gamma)}}{\partial{d}} = & \ \gamma'\log \gamma - \gamma'\log(1-\gamma) + \gamma'\log\left(2^n-d\right)- \gamma'\log d - \frac{\gamma}{d} + \frac{1-\gamma}{2^n-d} \\
 = & \ \gamma'\log\left(\frac{\gamma}{d}\right) - \gamma'\log\left(\frac{1-\gamma}{2^n-d}\right) - \left(\frac{\gamma}{d} - \frac{1-\gamma}{2^n-d} \right) \\
 = & \ \gamma'\left(\log p-\log q\right) - ( p -  q)
\end{align}
From this we obtain the first identity. Taking the second derivative with respect to $d$ we get
\begin{align}
 0 = & \gamma''( \log p - \log q ) + \gamma'\left( p'/p - q'/q \right) - (p' - q') \\
   = & \gamma''( \log p - \log q ) + d\left(p'\right)^2 / p + \left(2^n-d\right)\left(q'\right)^2/q  
\end{align}
Clearly, $\gamma'' < 0$. 
\end{proof}

\begin{claim}
For every $d\in \left[1,2^k\right]$ we have $p>q$. 
\end{claim}
\begin{proof}
Suppose that $p=q$ for some $d$. Then $p=q=2^{-n}$ what contraddicts to the first equation. 
Since $p$ and $q$ are continous with respecto to $d$, we have either $p>q$ or $p<q$. The first holds for $d=2^{k}$.
\end{proof}

\begin{claim}
Suppose that $k < c n$ for some sufficiently small absolute constant $c$. Then we have 
\begin{equation}\label{eq:Shannon_AuxiliaryFunction_DerivativeAsimptotic}
 \gamma'(d) = \bigO{ \frac{n-k}{d(n-\log d)^2} }
\end{equation}
\end{claim}
\begin{proof}
Recall, that the number $\gamma(d)$ is a solution of the equation $F(\gamma(d)) = -k$ where $F$ is a function defined by equation
(\ref{eq:ShanonEntropy_ExtremeDistributions_AuxiliaryFunction}). This equation may be rewriten as 
\begin{equation}\label{eq:Shannon_AuxiliaryFunction_v2}
 \gamma = \frac{\log \left(2^n-d\right) - k + \mathbf{H}(\gamma)}{\log\left(2^n-d\right) - \log d}
\end{equation}
where $\mathbf{H}(\gamma) = -\gamma\log \gamma -(1-\gamma)\log (1-\gamma)$ is the Shannon Entropy of a random variable taking two values with probabilities $\gamma$ and $1-\gamma$ respectively. 
Since $d \leqslant 2^{k}\leqslant 2^{n-2}$, we have the following estimates
\begin{align}
 \log\left(2^{n}-d\right)-\log d \geqslant n - \log d -1 \geqslant \frac{n-\log d}{2} \label{ineq:estimate1}\\
 1\geqslant 1-2^{-n}d \geqslant \frac{1}{2} \label{ineq:estimate2}
\end{align}
Thus, by (\ref{ineq:estimate1}), (\ref{ineq:estimate2}) and the fact that $\mathbf{H}(\gamma)\in [0,1]$, we get
\begin{align}\label{ineq:estimate3}
 \gamma(d) = \Theta\left( \frac{n-k}{n-\log d} \right)
\end{align}
Differentiating with respect to $d$ at a point $d<2^k$ we obtain
\begin{equation}
 \gamma'(d) = \log(e)\frac{\left(2^n-d\right) \log \left(2^n-d\right)+d \log (d)-k
   2^n + 2^{n}\mathbf{H}(\gamma)}{d \left(2^n-d\right) \left(\log
   \left(2^n-d\right)-\log (d)\right)^2} - \frac{\gamma'(d)\log\left( \frac{\gamma(d)}{1-\gamma(d)} \right)}{\log
   \left(2^n-d\right)-\log (d)}
\end{equation}
From the inequalities (\ref{ineq:estimate1}) and (\ref{ineq:estimate2}) it follows that the first term in the expression above is equal to $\Theta\left( \frac{n-k}{d(n-\log d)^2} \right)$. Now we will estimate the second term.
Consider the case $\gamma(d) < \frac{1}{2}$. Then
\begin{align}
 \left| \frac{\log\left( \frac{\gamma(d)}{1-\gamma(d)} \right)}{\log
   \left(2^n-d\right)-\log (d)} \right| \leqslant \left| \frac{\log\left( \frac{\gamma(1)}{1-\gamma(1)} \right)}{\log
   \left(2^n-2^k\right) - k}  \right| = \frac{\log \left(\frac{n-k}{k} \right) + \bigO{1}}{n-k -1} = \bigO{\frac{\log n}{ n-k}}
\end{align}
where we have used the fact that $\gamma(1) = \frac{n-k}{n} + \bigO{\frac{1}{n}}$ implied by (\ref{eq:Shannon_AuxiliaryFunction_v2}), and the assumption $k\leqslant n-1$. 
If $\gamma(d) > \frac{1}{2}$ then the second term is negative. 
Thus, provided that $k < c n$ for sufficiently small constant $c$, the result follows.
\end{proof}

\begin{claim}
We have $ \mathbf{H}_{1}^{\textup{Metric},\{0,1\},0}(X) = k $.
\end{claim}
\begin{proof}
Define for the convinience $\gamma(0) = 0$. Observe, that the numbers $\gamma(d) $ are increasing and since $\gamma\left(2^k\right)=1$ we have $\sum\limits_{x}\PrD{X}(x) = 1$ (a telescopic sum). Therefore we have indeed defined a probability measure. Let $D$ be any boolean function on $\{0,1\}^n$ such that $d = |D| < 2^k$. 
Since $\gamma(d)$ is concave then $\PrD{X}\left(x_d\right) = \gamma(d) - \gamma(d-1)$ is decreasing with $d$.  
Therefore
\begin{align}
 \mathbf{E}D(X) = & \sum_{x}\PrD{X}(x) \cdot D(x) \\
 \leqslant & \max\limits_{|I| = d} \sum\limits_{i \in I}^{d}\PrD{X}\left(x_i\right) \\
 \leqslant & \sum\limits_{i=1}^{d} \PrD{X}\left(x_i\right) \\
 = & \gamma_{d} = p(d)\cdot |D|
\end{align}
and by the chatacterization in Lemma, the first part follows. 
\end{proof}

\begin{claim}
We have $\mathbf{H}_1(X) = \bigO{ k^2 n^{-1} + k n^{-1} \log n }$. 
\end{claim}
\begin{proof}
Now we estimate the entropy of $X$. By definition
\begin{align}
 \mathbf{H}\left(X\right) = & - \gamma(1)\log \gamma(1) - \sum\limits_{d=2}^{2^k}\left(\gamma(d)-\gamma(d-1)\right)\log \left(\gamma(d)-\gamma(d-1)\right) 
\end{align}
Since $\gamma$ is concave, we have $\gamma(d) - \gamma(d-1) \leqslant \gamma'(d-1)$.
The function $t\rightarrow -t\log t$ is increasing for $t\leqslant \frac{1}{2}$ and for $d\geqslant 2$ and sufficiently large $n$, by concavity again we have $\gamma'(d-1) \leqslant \gamma'(1) \leqslant \frac{1}{2}$. Hence,
\begin{align}
 \mathbf{H}\left(X\right) \leqslant & - \gamma(1) \log \gamma(1) - \sum\limits_{d=1}^{2^k-1} \gamma'(d) \log  \gamma'(d) 
\end{align}
The function $d \rightarrow -\gamma'(d) \log \gamma'(d)$ is decreasing, as $\gamma'(d)$ decreases and $\gamma'(d)\leqslant \gamma'(1)\leqslant \frac{1}{2}$. Thus
\begin{align}\label{ineq:ShannonGap_EntropyEstimate}
 \mathbf{H}(X) \leqslant & \ \gamma(1) \log \gamma(1) - \gamma'(1)\log\gamma'(1) -\int\limits_{1}^{2^{k}-1} \gamma'(d) \log \gamma'(d) \ \mbox{d} d \nonumber \\
 \leqslant & \ 1-\int\limits_{1}^{2^k} \gamma'(d) \log \gamma'(d) \ \mbox{d} d
\end{align}
Using the estimate (\ref{eq:Shannon_AuxiliaryFunction_DerivativeAsimptotic}), for some constant $C > 1$ we obtain
\begin{align}\label{ineq:ShannonGap_IntegralEstimate1}
 -\int\limits_{1}^{2^k} \gamma'(d) \log \gamma'(d) \ \mbox{d} d \leqslant & -\int \limits_{1}^{2^k} \frac{C(n-k)}{d(n-\log d)^2}\log \left( \frac{C(n-k)}{d(n-\log d)^2} \right) \mbox{d} d + \bigO{C} \\
 = & -C \int\limits_{1}^{2^k} \frac{n-k}{d(n-\log d)^2}\log \left( \frac{n-k}{d(n-\log d)^2} \right) \mbox{d} d + \bigO{ C\log C }
\end{align}

Integrating and using the inequality $\ln (1+x) \leqslant x$ for $x > -1$, we get
\begin{align}\label{ineq:ShannonGap_IntegralEstimate2}
 -\int\limits_{1}^{2^k} \frac{n-k}{d(n-\log d)^2}\log \left( \frac{n-k}{d(n-\log d)^2} \right) \mbox{d} d = &
\ln 2 \cdot k + \ln^2 2 \cdot (n-k)\log \left(\frac{n-k}{n}\right) + \ln 2 \cdot \log (n-k) \nonumber \\ & +\ln 2 \cdot \frac{n-k}{n}\log \frac{n-k}{n^2} + \frac{2k}{n} \nonumber \\
= & \ln 2 \left( k + (n-k)\ln \left(1-\frac{k}{n}\right) \right) + \ln\left(1-\frac{k}{n}\right) + \frac{k}{n}\log n \nonumber \\ & + \left(1-\frac{k}{n}\right)\log\left(1-\frac{k}{n}\right)+\frac{2k}{n} \nonumber \\
\leqslant & (1+\ln 2)\cdot k^2 n^{-1} + k n^{-1}\log n
\end{align}
Finally, inequalities (\ref{ineq:ShannonGap_EntropyEstimate}), (\ref{ineq:ShannonGap_IntegralEstimate1}) and (\ref{ineq:ShannonGap_IntegralEstimate2}) yield the result.
\end{proof}
\noindent The proof follows by claims.
\end{proof}
\noindent This result directly implies the following one:
\begin{corollary}\label{ShannonMetric_vs_Shannon_Separation}
For some absolute constant $c>0$ and every sufficiently large $n$ there exists a random variable such that $\mathbf{H}^{\textup{Metric},\det\{0,1\},\epsilon}_1(X) = c n$ but $\mathbf{H}^{}_1(X) \leqslant c n / 2$. 
\end{corollary}
\noindent Now we give separation between randomized and deterministic distinguishers for the Shannon Entropy:
\begin{corollary}
For some absolute constant $c \in (0,1)$, for every $n$ there exists a random variable $X\in\{0,1\}^n$ such that:
\begin{enumerate}[label = (\roman{*})]
 \item Metric Shannon Entropy of $X$ is $k \geqslant c n$, against all deteterministic boolean functions and $\epsilon = 0$
 \item Metric Shannon Entropy of $X$ is $k - \bigOm{n}$, against all randomized circuits and $\epsilon = \bigOm{1}$.
\end{enumerate} 
\end{corollary}
\begin{proof}
\noindent We will make use of the following result, which says that Shannon Entropy is continuous (almost Lipschitz) with respect to the statistical distance. The proof is technical and is given in the Appendix.
\begin{lemma}\label{ShannonEntropy_Lipchitz}
Let $X,Y\in\{0,1\}^n$ be random variables. Then $\left| \mathbf{H}_1(X) - \mathbf{H}_1(Y)\right| = \bigO{n\Delta(X,Y)}-2\Delta(X,Y)\log\Delta(X,Y)$.
\end{lemma}
\begin{corollary}\label{DifferentShannon_Distinguishing}
Let $X,Y\in\{0,1\}^n$ be random variables such that  $\mathbf{H}_1\left(Y\right) - \mathbf{H}_1(X) = d \geqslant 1$. Then
\begin{equation*}
 \Delta(X,Y) \geqslant \Omega(d/n)
\end{equation*}
\end{corollary}
\begin{proof}
Let $\epsilon = \Delta(X,Y)$. Lemma \ref{ShannonEntropy_Lipchitz} gives us $\left| \mathbf{H}_1(X) - \mathbf{H}_1(Y)\right| < c n \epsilon + 2 \epsilon\log (1/\epsilon)$ for some absolute constant $c$. If $2\log (1/\epsilon) > cn$ then $\epsilon \leqslant 2^{-cn/2}$ and for sufficiently large $n$ we get $1\leqslant d\leqslant 4\epsilon \log (1/\epsilon) \leqslant 4\cdot 2^{-cn/2}\cdot (cn/2) < 1$. Hence we must have $\left| \mathbf{H}_1(X) - \mathbf{H}_1(Y)\right| \leqslant 2cn\epsilon$ for large $n$. For the remaining (finitely many) cases $n=1,\ldots,N=N(c)$ for every $n$ we find a number $\gamma_n$ such that $\left| \mathbf{H}_1(X) - \mathbf{H}_1(Y)\right| \leqslant \gamma_n\Delta(X,Y)$, under the constraint $1\leqslant \left|\mathbf{H}_1(X) - \mathbf{H}_1(Y)\right|$. By a compactness argument $\gamma_n$ are well defined and for the number $\gamma = \max\left(\gamma_1,\ldots,\gamma_{N}, 2c\right)$ we have $\left| \mathbf{H}_1(X) - \mathbf{H}_1(Y)\right| < \gamma\Delta(X,Y)$ for all $n$. Especially, $\Delta(X,Y) > \gamma^{-1}\left| \mathbf{H}_1(X) - \mathbf{H}_1(Y)\right|$, provided that $\left| \mathbf{H}_1(X) - \mathbf{H}_1(Y)\right| \geqslant 1$.
\end{proof}
\noindent 
Let $X$ be distribution from Corollary \ref{ShannonMetric_vs_Shannon_Separation}. Consider the set $\cY$ of all distributions $Y\in\{0,1\}^n$ with Shannon Entropy at least $\frac{3}{4}c n$. 
By Corollary \ref{DifferentShannon_Distinguishing} we obtain that for every distribution 
$\PrD{Y} \in \cY$ there exists a $[0,1]$-valued function $D$ such that
$\mathbf{E}D(X) - \mathbf{E}D(Y) \geqslant \Omega(1)$. But it means that $\mathbf{H}_1^{\text{HILL}\det [0,1],\Omega(1)} \leqslant \frac{3}{4}c n$. Since there is no restriction on the complexity, the same holds for Metric entropy. Since for unbounded circuits, Metric Entropy against $[0,1]$-valued and boolean randomized distinguishers is the same (up to a
arbitrary small absolute error), the result follows.
\end{proof}

\begin{replemma}{HighEntropy_LargeSupport}
Let $X\in\{0,1\}^n$ be a random variable satisfying $\mathbf{H}_{\alpha}(X)\geqslant k$.
Then $\left| \mathrm{supp}\left(\PrD{X}\right) \right| \geqslant 2^k$. 
\end{replemma}
\begin{proof}
Suppose that the distribution of $X$ is supported on some set $S$. The entropy constraint yields
\begin{equation}
 \sum\limits_{x\in S}\left(\PrD{X}(x)\right)^{\alpha} \leqslant 2^{-(\alpha-1)k}
\end{equation}
on the other side, the Jensen inequality gives us
\begin{equation}
 |S|^{-\alpha} = \left(\frac{1}{|S|}\sum\limits_{x\in S}\PrD{X}(x) \right)^{\alpha} \leqslant \frac{1}{|S|}\sum\limits_{x\in S}\left(\PrD{X}(x)\right)^{\alpha}
\end{equation}
From these two inequalities it follows that $|S| \geqslant 2^k$.
\end{proof}

\begin{replemma}{ShannonEntropy_Lipchitz}
Let $X,Y\in\{0,1\}^n$ be random variables. Then 
\begin{equation*}
\left| \mathbf{H}_1(X) - \mathbf{H}_1(Y)\right| \leqslant \bigO{n\Delta(X,Y)}-2\Delta(X,Y)\log\Delta(X,Y). 
\end{equation*}
\end{replemma}
\begin{proof}
Suppose that distributions of $X,Y$ are chosen to maximize $\left| \mathbf{H}_1(X) - \mathbf{H}_1(Y)\right|$ under the constraint with $\Delta(X,Y) = \epsilon$. Assume that $\mathbf{H}_1(X) \leqslant k \leqslant \mathbf{H}_1(Y)$. 
Consider the sets $S^{-} = \left\{x:\, \PrD{Y}(x) < \PrD{X}(x) \right\}$ and $S^{+} = \left\{x:\, \PrD{Y}(x) > \PrD{X}(x) \right\}$. We can assume that they are nonempty as otherwise $\PrD{X} = \PrD{Y}$.
The proof is divided into claims and starts with the following useful inequality:
\begin{claim}\label{ineq:Shannon_TwoPoint_Monotone}
Let $H(p,q) = -p\log p -q\log q$. Suppose that $0 \leqslant p\leqslant q$ and $p+q\leqslant 1$. Then 
\begin{align}
 H(p+\epsilon,q-\epsilon) < &\ H(p,q), \quad -p \leqslant \epsilon < 0 \label{ineq:Shannon_TwoPoint_Monotone1} \\
 H(p+\epsilon,q-\epsilon) > &\ H(p,q), \quad 0 < \epsilon < q-p \label{ineq:Shannon_TwoPoint_Monotone2} \\
 H(p+\epsilon,q-\epsilon) < &\ H(p,q), \quad q-p < \epsilon \label{ineq:Shannon_TwoPoint_Monotone3}
\end{align}
\end{claim}
\noindent Next we derive several properties of $\PrD{X},\PrD{Y}$ over the sets $S^{+},S^{-}$.
\begin{claim}
The set $S^{-}$ contains only one element $x=x'$. 
\end{claim}
\begin{proof}
Suppose that $S^{-}$ contains two points $x_1,x_2\in S^{-}$, such that $\PrD{X}\left(x_1\right) \leqslant \PrD{X}\left(x_2\right)$. Consider a distribution $\PrD{X'}$ given by $\PrD{X'}\left(x_1\right) = \PrD{X'}\left(x_1\right)-\delta$, $\PrD{X'}\left(x_2\right) = \PrD{X'}\left(x_2\right)+\delta$ and $\PrD{X'}(x) = \PrD{X}(x)$ if $x\not\in\left\{x_1,x_2\right\}$ where $\delta$ is sufficiently small positive number (from the definition of $S^{-}$ we have $\PrD{X}(x) \in (0,1)$ for $x\in S^{-}$, provided that $S^{-}$ has at least two elements). Since
\begin{equation}\label{ineq:Shannon_TwoPoint_Monotone}
 -(a-\delta)\log(a-\delta) -(b+\delta)\log(b+\delta) < - a\log a - b\log b \quad \text{ for } 0 < a\leqslant b < 1,\ \delta > 0
\end{equation}
we have $\mathbf{H}_{1}(X') < \mathbf{H}_{1}(X)$. Since $\Delta(X,Y) = \Delta(X',Y)$ (for sufficiently small $\delta$) we get a contradiction with the choice of $X,Y$. Hence, we may assume that $\left|S^{-}\right| = 1$. 
\end{proof}
\begin{claim}
The distribution $\PrD{Y}$ is uniform over $S^{+}$. 
\end{claim}
\begin{proof}
We can assume $\left|S^{+}\right|>1$. Suppose that that $\PrD{Y}\left( x_1 \right) < \PrD{Y}\left(x_2\right)$ for $x_1,x_2\in S^{+}$. Considering a distribution $\PrD{Y'}$ given by $\PrD{Y'}\left(x_1\right) = \PrD{Y}\left(x_1\right)+\delta$, $\PrD{Y'}\left(x_2\right) = \PrD{Y}\left(x_2\right)-\delta$ and $\PrD{Y'}(x) = \PrD{Y}(x)$ if $x\not\in\left\{x_1,x_2\right\}$ (from the definition of $S^{+}$ we have $\PrD{Y}(x)\in(0,1)$ for $x\in S^{+}$ provided that $\left|S^{+}\right| > 1$), by (\ref{ineq:Shannon_TwoPoint_Monotone}) we obtain $\mathbf{H}_1(Y') > \mathbf{H}_1(Y)$. Since $\Delta\left(X,Y'\right) = \Delta(X,Y)$, we get a contradiction. 
\end{proof}
\begin{claim}
$\PrD{X}(x) > 0$ for at most one element $x=x''\in S^{+}$. 
\end{claim}
\begin{proof}
Suppose that $0 < \PrD{X}\left(x_1\right) \leqslant \PrD{X}\left(x_2\right)$ for two different points $x_1,x_2\in S^{+}$. Define $\PrD{X'}$ by $\PrD{X'}\left(x_1\right) = \PrD{X'}\left(x_1\right)-\delta$, $\PrD{X'}\left(x_2\right) = \PrD{X'}\left(x_2\right)+\delta$ and $\PrD{X'}(x) = \PrD{X}(x)$ if $x\not\in\left\{x_1,x_2\right\}$ for sufficiently small $\delta > 0$. Then $\Delta(X',Y) = \Delta(X,Y)$ and by (\ref{ineq:Shannon_TwoPoint_Monotone}) we get $\mathbf{H}_1(X') < \mathbf{H}_1(X)$. Therefore, there is at most one point $x\in S^{+}$ such that $\PrD{X}(x) > 0$. Observe however, that this cannot hold: from the definition of $S^{+}$ we have $\epsilon = \sum\limits_{x \in S^{+}}\left( \PrD{Y}(x) - \PrD{X}(x) \right)$ and the fact that $Y$ is uniform over $S^{+}$ yields $\epsilon = \left|S^{+}\right|\PrD{Y}(x) -\PrD{X}(x)$. This implies $\PrD{Y}(x) = \frac{\PrD{X}(x) + \epsilon}{\left|S^{+}\right|}$. But now, the definition of $S^{+}$ yields the inequality $\epsilon > \left(\left|S^{+}\right|-1\right) \PrD{X}(x)$ and then $\PrD{Y}(x) >1 $, a contradiction.
\end{proof}
\begin{claim}
We have $\PrD{Y}(x') = \PrD{X}(x') - \epsilon$ and $\PrD{Y}(x'') =  \frac{\PrD{X}(x'') + \epsilon}{\left|S^{+}\right|}$
\end{claim}
\begin{proof}
It is easy to see that $\sum\limits_{x\in S^{-}}\left(\PrD{X}(x) - \PrD{Y}(x)\right) = \sum\limits_{x\in S^{+}}\left(\PrD{Y}(x) - \PrD{X}(x)\right) = \Delta(X,Y)$. This immediately implies the first equality. Second is obtained because of the previous two claims.
\end{proof}
\begin{claim}
We have $\PrD{X}(x') \geqslant \PrD{X}(x'')+\epsilon$ 
\end{claim}
\begin{proof}
Otherwise, we have $\PrD{X}(x') < \PrD{X}(x'')+\epsilon$. Consider then a distribution $\PrD{X'}$ given by $\PrD{X'}\left(x'\right) = \PrD{X}\left(x'\right)-\epsilon$, $\PrD{X'}\left(x''\right) = \PrD{X}\left(x''\right)+\epsilon$ and $\PrD{X'}(x) = \PrD{X}(x)$ if $x\not\in\left\{x',x''\right\}$ (this is a probability distribution because $\PrD{X'}(x) = \PrD{X}(x') - \epsilon = \PrD{Y}(x') \geqslant 0$). Then $x=x''$ is the only point such that $\PrD{X'}(x) \geqslant \PrD{Y}(x)$. Thus $\Delta(X',Y) = \left|\PrD{X'}(x'')-\PrD{Y}(x'') \right| = \left| \PrD{X}(x'')-\PrD{Y}(x'') + \epsilon \right|$.
Observe now that the definition of $S^{+}$ implies $\PrD{X}(x'')-\PrD{Y}(x'') < 0$ and $\Delta(X,Y) \leqslant \epsilon$ implies 
$-\epsilon < \PrD{X}(x'')-\PrD{Y}(x'')$. Therefore, $\Delta(X',Y) \leqslant \epsilon$.
Finally, note that $\PrD{X'}(x') = \PrD{X}(x') - \epsilon < \PrD{X}(x'')$ and $\PrD{X'}(x'')=  \PrD{X}(x'') + \epsilon > \PrD{X}(x')$. If $\PrD{X}(x'') \leqslant \PrD{X}(x')$, then by (\ref{ineq:Shannon_TwoPoint_Monotone1}) we get
$H\left(\PrD{X'}(x'),\PrD{X'}(x'')\right) < H\left( \PrD{X}(x''),\PrD{X}(x')\right)$. Otherwise, 
$\PrD{X}(x') < \PrD{X}(x'')$ and since $\PrD{X'}(x') < \PrD{X}(x')$ and $\PrD{X'}(x'') > \PrD{X}(x'')$,  (\ref{ineq:Shannon_TwoPoint_Monotone1}) yields 
$H\left(\PrD{X'}(x'),\PrD{X'}(x'')\right) < H\left(\PrD{X}(x'),\PrD{X}(x'')\right)$. Anyway, we obtain $\mathbf{H}_1(X') < \mathbf{H}_1(X)$, a contradiction.
\end{proof}
Now we are in position to give the final estimate. We consider two cases: $\epsilon > 2^{-n}$ and $\epsilon < 2^{-n}$.
\begin{claim}
Suppose that $\left|S^{+}\right|\geqslant 2$. Then $\PrD{X}(x'') < \frac{\epsilon}{\left|S^{+}\right|-1}$.  
\end{claim}
\begin{proof}
The definition of $S^{+}$ implies that $\PrD{X}(x'') < \PrD{Y}(x'') =  \frac{\PrD{X}(x'') + \epsilon}{\left|S^{+}\right|}$. 
\end{proof}
\begin{claim}
For the case $\left|S^{+}\right| > 1$ we have $\mathbf{H}_1(Y) - \mathbf{H}_1(X) \leqslant 6\epsilon + n\epsilon - \epsilon\log\epsilon$. 
\end{claim}
\begin{proof}
Consider the case $\left|S^{+}\right| \geqslant 2$. Define then $\PrD{X'}$ as $\PrD{X}(x'') = 0$, $\PrD{X'}(x') = \PrD{X}(x') + \PrD{X}(x'')$, $\PrD{X'}(x) = \PrD{X}(x)$ if $x\not\in \{x',x''\}$. Note, that $\mathbf{H}_1(X') < \mathbf{H}_1(X)$ by (\ref{ineq:Shannon_TwoPoint_Monotone1}). Then we obtain
\begin{align}
 \mathbf{H}_1(Y) - \mathbf{H}_1(X) \leqslant & \ \mathbf{H}_1(Y) - \mathbf{H}_1(X') \\
 =  & \sum\limits_{x\in S^{-} \cup S^{+}} \left(\PrD{X'}(x)\log\PrD{X'}(x) - \PrD{Y}(x)\log\PrD{Y}(x) \right) = \\ 
 = & \ \PrD{X'}(x')\log\PrD{X'}(x') - \PrD{Y}(x')\log\PrD{Y}(x')  - \sum\limits_{x\in S^{+}}\PrD{Y}(x)\log \PrD{Y}(x) \\
 < & \ \PrD{X'}(x')\log\PrD{X'}(x') - \PrD{Y}(x')\log\PrD{Y}(x') - \epsilon \log \epsilon + \epsilon\log\left|S^{+}\right| 
\end{align}
Since $\Delta(X',Y) = \Delta(X,Y) + \PrD{X}(x'') \leqslant 2\epsilon$ and the function $t\rightarrow t\log t$ is convex, it follows that 
\begin{align}\label{ineq:x1}
 \PrD{X'}(x')\log \PrD{X}(x') - \PrD{Y}(x')\PrD{Y}(x') \leqslant & 
 \left|\PrD{X'}(x') - \PrD{Y}(x')\right| \left.\left( t\log t \right)'\right|_{t=\log \PrD{X'}(x')} \nonumber \\
 \leqslant &  2\epsilon \left(1/\ln 2 + \log \PrD{X'}(x') \right) \nonumber \\
 \leqslant & 6\epsilon
\end{align}
and since $\left|S^{+}\right| < 2^{n}$, the result follows.
\end{proof}

\begin{claim}
If $\left|S^{+}\right| = 1$ then $\mathbf{H}_1(Y) - \mathbf{H}_1(X) < -2\epsilon \log\epsilon + 2\epsilon $ 
\end{claim}
\begin{proof}
If $\left|S^{+}\right| = 1$ then we have
\begin{align}
 \mathbf{H}_1(Y) - \mathbf{H}_1(X) =  & \PrD{X}(x')\log \left(\PrD{X}(x')\right) - \PrD{Y}(x')\log \PrD{Y}(x') \nonumber \\
 & + \left(\PrD{X}(x'')\right)\log \left(\PrD{X}(x'')\right) - \PrD{Y}(x'')\log \PrD{Y}(x'')
\end{align}
In the same way as in (\ref{ineq:x1}), we prove that the first expressions is at most $3\epsilon$. Now we have to estimate the second one. If $\PrD{X}(x'')\geqslant \epsilon$ then we get
\begin{align}
  \PrD{X}(x'')\log \left(\PrD{X}(x'')\right) - \PrD{Y}(x'')\log \PrD{Y}(x'') \leqslant & 
 -\left|\PrD{X}(x'') - \PrD{Y}(x'')\right| \left.\left( t\log t \right)'\right|_{t=\log \PrD{X'}(x'')} \nonumber \\
 \leqslant &  -\epsilon \left(1/\ln 2 + \log \PrD{X}(x') \right) \nonumber \\
 < & -\epsilon \log\epsilon 
\end{align}
In turn, if $\PrD{X}(x'') < \epsilon$, then $\PrD{Y}(x'') = \PrD{X}(x'') + \epsilon < 2\epsilon$. Thus, provided that $\epsilon < 1/4$,
\begin{align}
 \PrD{X}(x'')\log \left(\PrD{X}(x'')\right) - \PrD{Y}(x'')\log \PrD{Y}(x'') < & -\PrD{Y}(x'')\log \PrD{Y}(x'') \\
 \leqslant & -2\epsilon \log (2\epsilon) 
\end{align}
and if $\epsilon > 1/4$, then trivially $-\PrD{Y}(x'')\log \PrD{Y}(x'') \leqslant 1/2 < 2\epsilon$. Summing up, for the case $S^{+}$, we have proven that $ \mathbf{H}_1(Y) - \mathbf{H}_1(X) < -2\epsilon\log \epsilon + 2\epsilon$.
\end{proof}
\noindent We are left with the problem if estimating $\mathbf{H}_1(Y) - \mathbf{H}_1(X)$ for the extremely small values of
$\epsilon$.
\end{proof}

\end{document}